\documentclass[english,11pt]{article}
\usepackage{graphicx}
\usepackage{color}
\usepackage[ruled]{algorithm2e}
\usepackage{comment}
\usepackage{lmodern}
\usepackage[T1]{fontenc}
\usepackage[a4paper]{geometry}
\usepackage[authoryear]{natbib}
\usepackage{babel}
\usepackage{float}
\usepackage{enumitem}
\usepackage{amsmath}
\usepackage{amsthm}
\usepackage{amssymb}
\usepackage{bbm}
\usepackage{appendix}
\usepackage[unicode=true,
 bookmarks=true,bookmarksnumbered=false,bookmarksopen=false,
 breaklinks=true,pdfborder={0 0 0},pdfborderstyle={},colorlinks=true]
 {hyperref}
\usepackage{cleveref}
\usepackage{xspace}
\usepackage{dsfont}
\usepackage[dvipsnames,svgnames,x11names,hyperref]{xcolor}
\hypersetup{colorlinks,linkcolor={RoyalBlue},urlcolor=RoyalBlue,citecolor=RoyalBlue}
\usepackage[normalem]{ulem}
\usepackage{caption}
\usepackage{subcaption}
\usepackage{authblk}
\usepackage{marginnote}
\usepackage{tikz} % for flow charts
\usetikzlibrary{shapes.geometric, arrows}
\tikzset{
  basics/.style={minimum width=20mm, minimum height=7.5mm, text centered, draw=black},
  startstop/.style={rectangle, rounded corners, basics, fill=red!30},
  io/.style={trapezium, trapezium left angle=70, trapezium right angle=110, basics, fill=blue!30},
  process/.style={rectangle, basics, fill=blue!30},
  decision/.style={ellipse,basics, fill=green!30},
  arrow/.style={thick,->,>=stealth},
}
\usetikzlibrary{bayesnet} % to create graphical 
\usetikzlibrary{calc}

\newcommand{\Cat}{\ensuremath{\text{Cat}}}
\newcommand{\Bernoulli}{\ensuremath{\text{Ber}}}
\newcommand{\Binomial}{\ensuremath{\text{Bin}}}
\newcommand{\Poi}{\ensuremath{\text{Poi}}}
\newcommand{\PB}{\ensuremath{\text{PoiBin}}}
\newcommand{\TP}{\ensuremath{\text{TransPoi}}}
\newcommand{\CB}{\ensuremath{\text{CondBer}}}
\newcommand{\SB}{\ensuremath{\text{SumBin}}}
\newcommand{\KL}{\ensuremath{\text{KL}}}
\newcommand{\bigo}{\mathcal{O}}

\newcommand{\E}{\mathbb{E}}

\theoremstyle{definition}

\newtheorem{proposition}{Proposition}[section]
\newtheorem{lemma}{Lemma}[section]
\theoremstyle{remark}

\title{Sequential Monte Carlo algorithms for agent-based models of disease transmission}
\author[1]{Nianqiao Ju \thanks{Corresponding author: nju@g.harvard.edu}}
\author[2]{Jeremy Heng}
\author[1]{Pierre E. Jacob}
\date{}
\affil[1]{Department of Statistics, Harvard University, USA}
\affil[2]{ESSEC Business School, Singapore}

\begin{document}
\maketitle

\begin{abstract}
  Agent-based models of disease transmission involve stochastic rules that
  specify how a number of individuals would infect one another, recover or be
  removed from the population. Common yet stringent assumptions stipulate
  interchangeability of agents and that all pairwise contact are equally
  likely.  Under these assumptions, the population can be summarized by
  counting the number of susceptible and infected individuals, which greatly
  facilitates statistical inference.  We consider the task of inference without
  such simplifying assumptions, in which case, the population cannot be
  summarized by low-dimensional counts.  We design improved particle filters,
  where each particle corresponds to a specific configuration of the population
  of agents, that take either the next or all future observations into account
  when proposing population configurations.  Using simulated data sets, we
  illustrate that orders of magnitude improvements are possible over bootstrap
  particle filters.  We also provide theoretical support for the approximations
  employed to make the algorithms practical. 
\end{abstract}

% \tableofcontents
% \pagebreak

\section{Introduction \label{sec:intro}}
\subsection{Statistical inference for agent-based models}
Agent-based models also called individual-based models are used in many fields, 
such as social sciences \citep{epstein2006generative}, demographics \citep{hooten2020statistical},
ecology \citep{deangelis2018individual} and macroeconomics \citep{turrell2016agent}. 
These models describe the time evolution of a population according to a Markov process. 
The population comprises of $N\in\mathbb{N}$ agents who
interact with one another through a probabilistic mechanism. 
The popularity of these models could be attributed to the ease of model building and interpretation, 
while allowing the representation of complex phenomena. 
Specialized software is available to facilitate their simulation and the visualization of their output
\citep[e.g.][]{tisue2004netlogo}.  

The primary use of these models seems to involve ``forward simulations'' using
various parameters corresponding to hypothetical scenarios.  In comparison,
relatively fewer works have tackled the question of statistical inference for
such models, i.e. the estimation of model parameters given available
observations, also known as model calibration.  Broad articles on the question
of fitting agent-based models include
\citet{grazzini2017bayesian,hooten2020statistical,hazelbag2020calibration} and
this is a very active research area.  The question of estimation is well-posed:
by viewing agent-based models as a subclass of hidden Markov models, we can
define the associated likelihood function, from which maximum likelihood and
Bayesian procedures can be envisioned. 

Gibbs sampling or ``data augmentation'' Markov chain Monte Carlo (MCMC) methods, that alternate between 
parameter and latent agent states updates are generically applicable, 
see \citet{fintzi2017efficient,bu2020likelihood} in the context of disease transmission. 
However, the mixing properties of these chains can be problematic,
as we illustrate in Appendix \ref{sub:mcmc}. 
Likelihood-based inference that avoid the use of 
data augmentation requires one to marginalize out the latent population process.  
This is computationally challenging as the number of possible configurations of the 
population of agents grows exponentially in $N$. 
Our contribution is to design improved particle filters to estimate the likelihood function of 
agent-based models, which can then be used for parameter inference procedures 
such as simulated maximum likelihood and particle MCMC \citep{andrieu2010particle}. 
Particle MCMC methods, based on more standard particle filters, have been used for
agent-based models before, see e.g. \citet{kattwinkel2017bayesian}. 
%A baseline approach is to rely on the literature on statistical inference for generative
%models when the likelihood function is intractable 
%\citep{diggle1984monte,gourieroux1993indirect,beaumont2002approximate}, which
%applies to agent-based models. 
For indirect inference and approximate Bayesian computation applied
to agent-based models we refer to % can be a useful and generic solution 
\citep{platt2020comparison,van2016predicting,chen2017bayesian,siren2018assessing}.

\subsection{Statistical inference for agent-based models of disease transmission}
\label{subsec:sis_intro}
We limit the scope of this article to some agent-based models of disease transmission. 
We next introduce a specific model which will be used to draw clear 
connections between agent-based models and more tractable but restrictive formulations 
of disease outbreak models. 

The state of agent $n$ at time $t$, denoted by $X_t^n$, takes the value $0$ or $1$,  
corresponding to a ``susceptible'' or ``infected'' status in the context of disease transmission. 
We consider a closed population of size $N$, $T\in\mathbb{N}$ time steps and a stepsize of $h>0$. 
Initial agent states follow independent Bernoulli distributions, i.e. 
\begin{equation}\label{eqn:static_X}
	X_0^n \sim \Bernoulli(\alpha_0),
\end{equation}
for $n\in[1:N]=\{1,\ldots,N\}$, where we assume for the time being that each agent has the same probability 
$\alpha_0\in(0,1)$ of being infected. 
The state of the population $X_t=(X_t^n)_{n\in[1:N]}\in\{0,1\}^N$ evolves according to 
a Markov process that depends on the interactions between the agents. 
The Markov transition specifies 
that, given the previous states $x_{t-h}\in\{0,1\}^N$, the next agent states 
$(X_t^n)_{n\in[1:N]}$ are conditionally independent Bernoulli variables 
with probabilities 
$(\alpha^n(x_{t-h}))_{n\in[1 : N]}$ given by 
\begin{align}
	\label{eqn:sis_alpha_homo}
	\alpha^n(x_{t-h}) = 
	\begin{cases}
      h \lambda N^{-1}\sum_{m=1}^N x_{t-h}^m,\quad & \mbox{if }x_{t-h}^n=0,\\
		1- h \gamma,\quad & \mbox{if }x_{t-h}^n=1.
	\end{cases}
\end{align}
Here $\lambda\in(0,1)$ and $\gamma\in(0,1)$ represent 
infection and recovery rate parameters respectively. 
An uninfected agent becomes infected with probability proportional
to $\lambda$ and to the proportion of infected individuals. 
An infected agent recovers with probability proportional to $\gamma$. 
The ratio $R_0 = \lambda / \gamma$ is known as the basic reproductive number, 
which characterizes the transmission potential of a disease~\citep{lipsitch2003transmission,britton2010stochastic}. 
The latent population process $(X_t)$ can be related to aggregated observations $(Y_t)$ 
of the entire population. For example, the number of infections reported at each time 
could be modelled as a Binomial distribution, i.e. 
\begin{equation}\label{eqn:static_Y}
	Y_t \mid X_t = x_t \sim \Binomial(I(x_t), \rho),
\end{equation}
where $I(x_t)=\sum_{n=1}^N x_t^n\in[0:N]$ counts the number of infected individuals  
and $\rho\in(0,1)$ represents a reporting rate. 
Equations \eqref{eqn:static_X}, \eqref{eqn:sis_alpha_homo} and \eqref{eqn:static_Y} define 
a susceptible-infected-susceptible (SIS) hidden Markov model. 

As the latent state space has cardinality $2^N$, evaluating the likelihood function 
with the forward algorithm would be impractical. 
Statistical inference is still feasible as the latent process admits low-dimensional summaries. 
Since \eqref{eqn:sis_alpha_homo} assumes that all agents are interchangeable 
and all pairwise contact are equally likely, the dynamics depends only on the 
number of infected individuals $I_t=I(X_t)$ and susceptible individuals $S_t=N-I_t$. 
More precisely, the Markov transition for each $t>0$ can be summarized as 
\begin{align}\label{eq:sis:binomialmodel}
  S_t = S_{t-h} + N^S_t - N^I_t, \qquad
  I_t  = I_{t-h} - N^S_t + N^I_t,
\end{align}
where $N^S_t \sim \Binomial(I_{t-h},h\gamma)$ and $N^I_t \sim \Binomial(S_{t-h},h \lambda N^{-1} I_{t-h})$ 
denote the number of new recoveries and infections, respectively.  
The initialization in \eqref{eqn:static_X} corresponds to having 
$I_0 \sim \Binomial(N,\alpha_0)$ and $S_0 = N-I_0$. 
The process $(S_t,I_t)$ is a Markov chain on the lower-dimensional state space $\{n,m\in[0:N]:n+m=N\}$, 
and not in $\{0,1\}^N$. 
When combined with \eqref{eqn:static_Y}, the resulting hidden Markov model can be estimated 
using particle filtering strategies; see e.g. \citet{ionides2006inference,endo2019introduction,whiteley2020inference}.

For large population sizes $N$, one can exploit the asymptotic properties of 
the dynamics \eqref{eq:sis:binomialmodel} to further simplify inference procedures. 
As $N\rightarrow\infty$, the finite population proportions $\bar{S}_t^N=S_t/N$ and 
$\bar{I}_t^N=I_t/N$ admit deterministic limits $\bar{S}_t$ and $\bar{I}_t$, defined by 
the recursion 
\begin{align}\label{eq:sis:detmodel}
	\bar{S}_t = \bar{S}_{t-h} + h\gamma\bar{I}_t - h\lambda\bar{S}_t\bar{I}_t,\qquad
	\bar{I}_t = \bar{I}_{t-h} - h\gamma\bar{I}_t + h\lambda\bar{S}_t\bar{I}_t,
\end{align}
with initial condition $(\bar{S}_t,\bar{I}_t)=(1-\alpha_0,\alpha_0)$. 
If fine time resolutions are desired, one can also consider the limit $h\rightarrow 0$, 
in which case $(\bar{S}_t,\bar{I}_t)$ converges to a deterministic continuous-time process 
$(\bar{S}(t),\bar{I}(t))$, defined by the following system of ordinary differential equations 
\begin{align}\label{eq:sis:odemodel}
	\frac{d\bar{S}(t)}{dt} = \gamma\bar{I}(t)-\lambda\bar{S}(t)\bar{I}(t),\qquad
	\frac{d\bar{I}(t)}{dt} = -\gamma\bar{I}(t)+\lambda\bar{S}(t)\bar{I}(t),
\end{align}
with $(\bar{S}(0),\bar{I}(0))=(1-\alpha_0,\alpha_0)$. 
We refer to \citet{allen2000comparison} and \citet[Chapter 3]{allen2008mathematical} 
for formal links between stochastic and deterministic SIS models,
in the form of limit theorems as $N\to \infty$. 
The dynamics in \eqref{eq:sis:detmodel} or \eqref{eq:sis:odemodel} 
combined with an observation model such as \eqref{eqn:static_Y} 
leads again to a fairly tractable statistical model, which can be estimated by 
maximum likelihood or in the Bayesian framework \citep[e.g.][]{osthus2019dynamic}.

The tractability in \eqref{eq:sis:binomialmodel} and its limits \eqref{eq:sis:detmodel}-\eqref{eq:sis:odemodel} 
hinges critically on the assumption that all agents are interchangeable and equally likely to be in contact.
These assumptions are strong as it would be desirable to let the infection and recovery rates 
depend on individual characteristics or features, and the structure of interactions be specified by 
an undirected network. Adopting any one of these generalizations would result in an 
agent-based model that cannot be summarized by low-dimensional counts. 
This article is concerned with the task of inference when these assumptions are relaxed. 
Our approach is to return to the hidden Markov model formulation with latent space $\{0,1\}^N$, 
and seek novel sampling algorithms that have polynomial rather than exponential costs in $N$. 
The guiding principle of our contribution is that we can improve the
performance of Monte Carlo algorithms, sometimes dramatically, by tailoring
them to the models at hand, and by using available observations.

The rest of this article is organized as follows. 
In Section \ref{sec:static}, we consider a ``static'' case with a single time step 
to illustrate some of the challenges and proposed solutions in a simple setting.
We describe how the cost of likelihood evaluation and sampling the posterior of agent states 
given observations can be reduced from exponential to polynomial in $N$.
In Section \ref{sec:sis}, we consider the dynamic setting with multiple time steps 
and design auxiliary particle filters and controlled SMC algorithms 
tailored to heterogeneous SIS agent-based models. 
We also provide theoretical support for the approximations employed to make 
the algorithms practical. 
We extend our methodology to susceptible-infected-recovered (SIR) models in Section \ref{sec:sir} 
and highlight the associated difficulties. 
Section \ref{sec:discussion} summarizes our findings and open questions.  
% Each section contains numerical experiments on toy examples, with $N$ ranging
% in the order $10^2-10^3$ and the time horizon $T$ in $10^0-10^2$.  
This manuscript represents preliminary work and further numerical experiments will be
added in a later version.
The code is available at \url{https://github.com/nianqiaoju/agents}. The proofs and more details are in appendices. 

\section{Static agent-based model}
\label{sec:static}
For simplicity we begin with the static model $X^n \sim \Bernoulli(\alpha^n)$ 
independently for $n\in[1:N]$, and allow each agent to be unique by 
modeling $\alpha^n=(1+\exp(-\beta^\top w^n))^{-1}$, 
where $\beta\in\mathbb{R}^d$ are parameters and $w^n\in\mathbb{R}^d$ 
are the covariates of agent $n$. 
This allows individual factors to account for the probability of infection $\alpha=(\alpha^n)_{n\in[1:N]}$.
Given an observation $y \in [0:N]$ modelled as \eqref{eqn:static_Y}, 
our inferential objective might include estimating the unknown parameters 
$\theta=(\beta,\rho)\in \Theta$ 
and/or the latent states $X=(X^n)_{n\in[1:N]}\in\{0,1\}^N$. 
The complete data likelihood is 
$p_{\theta}(x,y) = f_\theta(x) g_\theta(y | x)$, where 
\begin{equation} \label{eqn:static_complete}
f_\theta(x) = \prod_{n=1}^N \Bernoulli(x^n;\alpha^n),\quad 
g_\theta(y | x) = \Binomial(y;I(x), \rho) \mathbbm{1}(I(x) \geq y).
\end{equation}
Here $\Bernoulli(x^n;\alpha^n)$ and $\Binomial(y;I(x), \rho)$ denote the corresponding probability mass functions (PMF). 
Marginalizing over the latent states yields the likelihood $p_{\theta}(y)$. The agent states 
given the observation and a parameter follow the posterior distribution $p_{\theta}(x | y)$. 
In Sections \ref{sub:static_likelihood} and \ref{sub:static_posterior_states}, 
we examine the cost of exactly computing $p_{\theta}(y)$ and sampling from $p_{\theta}(x | y)$, 
respectively, and describe cheaper approximations. 
The gains are assessed numerically 
in Section \ref{sub:static_inference}. 

\subsection{Marginal likelihood computations}
\label{sub:static_likelihood}
A naive approach to compute the marginal likelihood $p_\theta(y)$ is to sum 
over all possible population configurations 
\begin{equation} \label{eqn:naive}
p_{\theta}(y) = \sum_{x \in \{0,1\}^N} p_{\theta} (x, y) 
= \sum_{x \in \{0,1\}^N} f_{\theta}(x)g_{\theta}(y | x).
\end{equation} 
This requires $\bigo(2^N)$ operations. 
A simple Monte Carlo approach involves sampling 
$X^{(p)}=(X^{(p),n})_{n\in[1:N]}\sim f_{\theta}$ independently for $p\in[1:P]$,
which represents $P$ possible configurations of the population, 
and return the Monte Carlo estimator 
$P^{-1}\sum_{p=1}^Pg_{\theta}(y | X^{(p)})$ that 
weights each configuration according to the observation density. 
The estimator is unbiased and only costs $\bigo(NP)$ to compute, 
but its variance might be prohibitively large for practical values of $P$,
depending on the observation model $g_\theta(y|x)$. 
Another issue with this estimator is that it can collapse to zero whenever 
$I(X^{(p)})<y$ for all $p\in[1:P]$. 
Following \citet{moral2015alive}, this can be circumvented by 
repeatedly drawing samples $X^{(p)}\sim f_{\theta}$ 
independently until there are $P$ configurations that satisfy 
$I(X^{(p)})\geq y$, 
and return the estimator ${(R-1)}^{-1}\sum_{p=1}^{R-1} g_{\theta}(y | X^{(p)})$, where $R\geq P$ denotes the number of repetitions needed. 
The resulting estimator can be shown to be unbiased and has 
a random cost of $\bigo(N\,\E[R])$, that depends on the value of $y$.
An obvious shortcoming of the above estimators is that the agents are sampled from $f_\theta$ without
using the available observation $y$.

We can in fact reduce the cost of computing \eqref{eqn:naive} without 
resorting to Monte Carlo approximations.  
The starting observation is that $I = I(X)$ under $X\sim f_\theta$ is the sum of 
$N$ independent Bernoulli random variables with non-identical success 
probabilities, and follows a distribution called ``Poisson Binomial''. 
We will refer to this distribution as \PB($\alpha$)
and write its PMF as
\begin{equation}\label{eqn:poibin}
	\PB(i;\alpha) = \sum_{x \in \{0,1\}^N}  \mathbbm{1}\left(\sum_{n=1}^N x^n = i\right)  \prod_{n=1}^N(\alpha^n)^{x^n} (1-\alpha^n)^{(1-x^n)}, 
	\quad i\in [0:N].
\end{equation}
Exact evaluation of \eqref{eqn:poibin} has been considered in 
\citet{barlow1984computing,chen1994weighted,chen1997statistical,hong2013computing} 
using different approaches. 
Defining $q(i,n)=\PB(i;\alpha^{n:N})$ for $i\in[0:N]$ and $n\in[1:N]$, we will employ the following recursion 
\begin{equation}\label{eqn:q-recursion}
	q(i,n) = \alpha^n q(i-1,n+1) + (1 - \alpha^n) q(i,n+1), \quad i\in[1:N],n\in[1:N-1],
\end{equation}
with initial conditions $q(0,n)=\prod_{m=n}^N(1-\alpha^{m})$ for $n\in[1:N]$, $q(1,N)=\alpha^N$ and $q(i,N)=0$ for $i\in[2:N]$.
The desired PMF $q(i,1)=\PB(i;\alpha)$ for $i\in[0:N]$ can thus be computed in $\bigo( N^2 )$ operations; 
see Appendix \ref{sec:recursive-poibin} for a derivation of \eqref{eqn:q-recursion}.

Using the above observation, we can rewrite the marginal likelihood as 
\begin{equation}\label{eqn:likelihood_simpler}
	p_{\theta}( y ) = \sum_{i=0}^N \PB(i;\alpha) \Binomial(y; i, \rho) \mathbbm{1}(i \geq y),
\end{equation}
which costs $\bigo( N^2 )$ to compute. Using a thinning argument detailed in Appendix \ref{sub:thinning}, 
the above sum is in fact equal to 
$p_{\theta}(y)=\PB(y;\rho \, \alpha)$. 
Although the marginal likelihood will not admit such tractability in the 
general setup considered in Section \ref{sec:sis}, 
the preceding observations will guide our choice of approximations. 

One can also rely on approximations of the Poisson Binomial distribution to 
further reduce the cost of computing \eqref{eqn:likelihood_simpler} to 
$\bigo( N )$. 
Choices include the Poisson approximation \citep{hodges1960poisson,barbour1984rate,wang1993number}, 
the Normal approximation \citep{volkova1996refinement} and 
the translated Poisson approximation \citep{barbour1999poisson,cekanavicius2001centered,barbour2002total}. 
We will focus on the translated Poisson approximation which 
exactly matches the mean and approximately the variance of a Poisson Binomial 
distribution. 
Let $\Poi(\lambda)$ denote a Poisson distribution with rate $\lambda>0$ and 
write the mean and variance of \PB($\alpha$) as $\mu=\sum_{n=1}^N\alpha^n$ and 
$\sigma^2=\sum_{n=1}^N\alpha^n(1-\alpha^n)$, respectively. 
The translated Poisson approximation of \eqref{eqn:poibin} is given by 
\begin{equation}\label{eqn:transpoi}
	\TP(i;\mu,\sigma^2) =
	\begin{cases}
	0, &i\in[0:\lfloor\mu-\sigma^2\rfloor-1],\\
	\Poi(i-\lfloor \mu-\sigma^2\rfloor;\sigma^2 + \{\mu-\sigma^2\}), &i\in[\lfloor\mu-\sigma^2\rfloor:N],
	\end{cases}
\end{equation}
where $\lfloor\mu-\sigma^2\rfloor$ and $\{\mu-\sigma^2\}$ denote the floor and 
fractional part of $\mu-\sigma^2$, respectively. 
Since $\mu$ and $\sigma^2$ can be computed in $\bigo( N )$ operations, 
the translated Poisson approximation~\eqref{eqn:transpoi} and the resulting 
approximation of \eqref{eqn:likelihood_simpler} only require $\bigo( N )$ 
operations. 
Hence this can be appealing in the setting of large population sizes 
$N$ at the expense of an approximation error that is well-studied. 
Indeed results in \citet[Theorem 2.1 \& Corollary 2.1]{cekanavicius2001centered} and 
\citet[Theorem 3.1]{barbour2002total} imply that the error, measured in 
the total variation distance, decay at the rate of $N^{-1/2}$ as $N\rightarrow\infty$.  

\subsection{Posterior sampling of agent states}
\label{sub:static_posterior_states}
Sampling from the posterior distribution $p_{\theta}(x | y)$ by 
naively enumerating over all $2^N$ 
configurations is computationally impractical. 
A key observation is that 
the observation density in \eqref{eqn:static_complete}  
depends on the high dimensional latent state $X\in\{0,1\}^N$ only through the 
one-dimensional summary $I(X)\in[0:N]$. 
This prompts the inclusion of $I = I(X)$ as an auxiliary variable.
Thus the joint posterior distribution can be decomposed as 
\begin{align}\label{eqn:static_augmented_posterior}
	p_{\theta}(x, i | y)= p_{\theta}(i | y) p_{\theta}(x | i).
\end{align}
The dominant cost of sampling the posterior distribution of the summary 
\begin{align}
	p_{\theta}(i | y) 
	= \frac{p_{\theta}(i)p_{\theta}(y|i)}{p_{\theta}(y)}
	=\frac{\PB(i;\alpha)\Binomial(y; i, \rho) \mathbbm{1}(i \geq y)}{p_{\theta}(y)},
	\quad i\in [0:N],
\end{align}
is the evaluation of the Poisson Binomial PMF \eqref{eqn:poibin}.
Recall from Section \ref{sub:static_likelihood} that 
this can be done exactly in $\bigo(N^2)$, and in $\bigo(N)$
using a translated Poisson approximation. 
The conditional distribution of the latent state given the summary 
\begin{align}
	p_{\theta}(x|i) = \frac{p_{\theta}(x)p_{\theta}(i|x)}{p_{\theta}(i)} 
	= \frac{\prod_{n=1}^N (\alpha^n)^{x^n}(1-\alpha^n)^{1-x^n}\mathbbm{1}(I(x)=i)}{\PB(i;\alpha)}
\end{align}
is known as a conditional Bernoulli distribution, which we will write as 
$p_{\theta}(x|i)=\CB(x;\alpha, i)$. 
Various sampling schemes have been proposed \citep{fan1962development,chen1994weighted}; see \citet[Section 4]{chen1997statistical} for an overview. 
We will rely on the sequential decomposition
\begin{align}\label{eqn:condber_seq}
	p_{\theta}(x|i) &= p_{\theta}(x^1|i) \prod_{n=2}^N p_{\theta}(x^n|x^{1:n-1},i)\\
	&=\prod_{n=1}^{N-1}\frac{(\alpha^n)^{x^n}(1-\alpha^n)^{1-x^n}q(i-i_{n-1}-x^n,n+1)}{q(i-i_{n-1},n)}
	\mathbbm{1}(x^N=i-i_{N-1}),\notag
\end{align}
where $i_0 = 0$ and $i_n = \sum_{m=1}^nx^m$ for $n\in[1:N-1]$. 
A derivation of \eqref{eqn:condber_seq} can be found in Appendix A of \citet{heng2020simple}. 
The values of $q(j,n)$ for $j\in[0:i], n\in[1:N]$ needed in
\eqref{eqn:condber_seq} can be 
precomputed using the same recursion as \eqref{eqn:q-recursion} in $\bigo(iN)$ cost. 
This precomputation is not necessary if these values are stored when  
computing Poisson Binomial probabilities.

To reduce the cost we can employ a Markov chain Monte Carlo (MCMC) method 
to sample from the conditional Bernoulli distribution $\CB(\alpha,i)$.
This incurs a cost of $\bigo(1)$ per iteration
and converges in $\bigo(N \log N)$ iterations, under some mild assumptions on $\alpha$,
as shown in \citet{heng2020simple}.
On a related note, we can design MCMC algorithms 
to target $p_{\theta}(x | y)$. These 
might for example update the state of a few agents 
by sampling from their conditional distributions
given every other variables, and alternately propose to swap 
zeros and ones in the vector $X\in\{0,1\}^N$. 
Each of these steps can be done with a cost independent of $N$,
but the number of iterations for the algorithm to converge
is expected to grow at least linearly with $N$.

\subsection{Numerical illustration}
\label{sub:static_inference}

We set up numerical experiments as
follows. 
We generate covariates $(w^n)$ from $\mathcal{N}(4,1)$ for $N=1000$ individuals 
independently, where $\mathcal{N}(\mu,\sigma^2)$ denotes a Normal distribution 
with mean $\mu\in\mathbb{R}$ and variance $\sigma^2>0$.  
Individual specific infection probabilities are computed as 
$\alpha^n = (1+\exp(-\beta  w^n))^{-1}$ using $\beta = 0.3$. 
We then simulate $X^n\sim \Bernoulli(\alpha^n)$ independently for all $n$, and sample
$Y\sim \Binomial(\sum_{n=1}^N X^n, \rho)$, with $\rho = 0.8$. We adopt a
Bayesian approach and assign an independent prior of 
$\mathcal{N}(0,1)$ on $\beta$ and $\text{Uniform}(0,1)$ on $\rho$. 
We focus on the task of sampling from the posterior 
distribution of $\theta = (\beta,\rho)$ given $y$ and the covariates $(w^n)$.

We consider random walk Metropolis--Hastings with exact likelihood calculation
(``MH-exact''), associated with a quadratic cost in $N$.  We also consider the
same algorithm with a likelihood approximated by a translated Poisson
(``MH-tp''), with a cost linear in $N$.  Finally we consider a pseudo-marginal
approach \citep{andrieu2009pseudo,andrieu2010particle} with likelihood
estimated with $P=20$ particles sampled from $f_\theta$ (``PMMH'').  Note that $P=20$
samples resulted in a variance of the log-likelihood estimates of
approximately $0.3$ at the data generating parameters (DGP). 
These samplers employ the same random walk, based on Normal proposals,
independently on $\beta$ and $\log(\rho/(1-\rho))$ with standard deviation of $0.2$. 
As a baseline we also consider a Gibbs sampler that alternates between the
updates of $\theta$ given $X$, employing the same proposal as above, and
updates of $X$ given $\theta$. These employ an equally weighted mixture of
kernels, performing either $N$ random swap updates, or a systematic Gibbs scan
of the $N$ components of $X$; thus the cost per iteration is linear in $N$. 

We first run ``MH-exact'' with $100,000$ MCMC iterations (after a burn-in of
$5,000$ iterations), to obtain estimates of the posterior means of $\beta$ and
$\rho$.  
Using these posterior estimates as ground truth, we compute the mean squared
error (MSE) of the posterior approximations obtained using each method with 
$20,000$ MCMC iterations (excluding a burn-in of $5000$) and 50 independent
repetitions. Table \ref{table:static} displays the MSE, as well as 
the relative wall-clock time to obtain each estimate.
Comparing the results of MH-exact and MH-tp shows that it is possible to 
save considerable efforts at the expense of 
small differences in the parameter estimates using a translated Poisson approximation.
The appeal of the PMMH approach compared to exact likelihood calculations is also clear. 
On the other hand, Gibbs samplers that alternate between the updates
of $\theta$ and $X$ do not seem to be competitive in this example. 

\begin{table}[htbp]
\centering
\begin{tabular}{c|cc|cc|c}
Method & \multicolumn{2}{c|}{$\mathbb{E}[\beta|y]$} & \multicolumn{2}{c|}{$\mathbb{E}[\rho|y]$} & Relative cost \\
& Bias$^2$ & Variance & Bias$^2$ & Variance & \\
\hline \hline
MH-exact & 25  & 93.3 & 0.74 & 6.39 & 128  \\
MH-tp & 22 & 52.3 & 0.32 & 2.83 & 1 \\
PMMH & 18 & 79.2 & 0.50 & 4.67 & 8 \\
Gibbs & 1040 & 113 & 91.1 & 2.85 & 42
\end{tabular}
\caption{Bias and variance when approximating the posterior mean of $\beta$ and $\rho$ 
with each inference method (in units of $10^{-4}$). 
The rightmost column displays the average cost for each method measured in 
terms of run-time, relative to MH-tp.
\label{table:static}}
\end{table}

\section{Susceptible-Infected-Susceptible model}
\label{sec:sis}
We now extend the agent-based SIS model in Section~\ref{subsec:sis_intro}. 
To allow for individual-specific attributes, for agent $n\in[1:N]$, 
we model her initial infection probability $\alpha_0^n$, 
infection rate $\lambda^n$ and recovery rate $\gamma^n$ as 
\begin{align}\label{eqn:latentprocess_parameters}
	\alpha_0^n = (1 + \exp(-\beta_0^\top w^n))^{-1}, \quad 
	\lambda^n = (1 + \exp(-\beta_{\lambda}^\top w^n))^{-1},\quad
	\gamma^n = (1 + \exp(-\beta_{\gamma}^\top w^n))^{-1}.
\end{align}
Here $\beta_0,\beta_{\lambda},\beta_{\gamma}\in\mathbb{R}^d$ are parameters 
and $w^n\in \mathbb{R}^d$ are the covariates of agent $n$.
The interactions between agents is assumed to be known and specified by 
an undirected network; inference of the network structure and extension to 
the time-varying case could be considered in future work. 
We will write $\mathcal{D}(n)$ and $\mathcal{N}(n)$ to denote the degree 
and neighbours of agent $n\in[1:N]$. 

For ease of presentation, we consider time steps of size $h=1$. 
The time evolution of the population is given by 
\begin{align}\label{eqn:dynamic_X}
	X_0\sim \mu_{\theta}, \quad  
	X_t | X_{t-1}=x_{t-1}\sim f_{\theta}(\cdot|x_{t-1}),\quad t\in[1:T]. 
\end{align}
The initial distribution $\mu_{\theta}(x_0) = \prod_{n=1}^N\Bernoulli(x_0^n;\alpha_0^n)$ 
corresponds to the static model in Section~\ref{sec:static} 
with infection probabilities $\alpha_0=(\alpha_0^n)_{n\in[1:N]}$. 
The Markov transition $f_{\theta}(x_t |x_{t-1})= \prod_{n=1}^N\Bernoulli(x_t^n;\alpha^n(x_{t-1}))$ 
has conditional probabilities $\alpha(x_{t-1})=(\alpha^n(x_{t-1}))_{n\in[1 : N]}$ given by 
\begin{align}\label{eqn:sis_alpha_hetero}
	\alpha^n(x_{t-1}) = 
	\begin{cases}
		\lambda^n \mathcal{D}(n)^{-1}\sum_{m\in\mathcal{N}(n)}x_{t-1}^m,\quad & \mbox{if }x_{t-1}^n=0,\\
		1-\gamma^n,\quad & \mbox{if }x_{t-1}^n=1,
	\end{cases}
\end{align}
for $n\in[1:N]$. We will assume that the cost of evaluating $\alpha$ is $\bigo(N)$. 
We refer readers to \citet{mcvinish_pollett_2012} for the asymptotic behaviour 
of this process as $N\to\infty$, in the case of a fully connected network.  
In \eqref{eqn:sis_alpha_hetero} we see the proportion of infected neighbours 
appearing in the infection probability at each time step. 
As an individual can have covariates $w^n$ that include measures 
of contact frequency, which would affect $\alpha^n$ via $\lambda^n$, 
it is possible for an individual with many neighbours (large $\mathcal{D}(n)$)
to have low infection probability if the frequency of contact is low.  

Equations \eqref{eqn:latentprocess_parameters}-\eqref{eqn:sis_alpha_hetero}
and the observation model in \eqref{eqn:static_Y} define a hidden Markov model 
on $\{0,1\}^N$, with unknown parameters $\theta=(\beta_0,\beta_{\lambda},\beta_{\gamma},\rho)\in
\Theta = \mathbb{R}^{3d}\times(0,1)$; see Figure~\ref{fig:graphical_model} for 
a graphical model representation in the case of a fully connected network. 
Given an observation sequence $y_{0:T} \in [0:N]^{T+1}$, the complete data likelihood is 
given by 
\begin{align}
	p_{\theta}(x_{0:T},y_{0:T}) = p_{\theta}(x_{0:T})p_{\theta}(y_{0:T}|x_{0:T}) 
	= \mu_{\theta}(x_0)\prod_{t=1}^Tf_{\theta}(x_t|x_{t-1})\prod_{t=0}^Tg_{\theta}(y_t|x_t).
\end{align}
Parameter inference will require marginalizing over the latent process
to obtain the marginal likelihood $p_{\theta}(y_{0:T})$ 
and estimation of agent states will involve sampling from the smoothing distribution
$p_{\theta}(x_{0:T}|y_{0:T})$. 
Exact computation of the marginal likelihood and marginals of the smoothing distribution 
using the forward algorithm and forward-backward algorithm, respectively, both require 
$\bigo(2^{2N}T)$ operations. 
As this is computationally prohibitive for large $N$, we will rely on sequential Monte Carlo (SMC) 
approximations. 

In Section~\ref{sub:sis_apf}, we describe how SMC methods 
can be used to approximate the marginal likelihood and 
the smoothing distribution. 
Like many Monte Carlo schemes, the efficiency of SMC crucially relies on the choice of proposal distributions. 
The bootstrap particle filter (BPF) \citep{gordon1993novel}, which corresponds to having the joint distribution 
of the latent process \eqref{eqn:dynamic_X} as proposal, can often give poor performance in practice when 
the observations are informative. 
By building on the insights from Section~\ref{sec:static}, 
we will show how the fully adapted auxiliary particle filter (APF) \citep{pitt1999filtering,carpenter1999improved} 
can be implemented. As the APF constructs a proposal transition at each time
step that takes the next observation in account, it often 
performs better than the BPF, although not always 
\citep{johansen2008note}. 
In Section~\ref{sub:sis_controlled_smc}, we adapt the ideas in \citet{guarniero2017iterated,heng2020controlled} to our setting 
and present a novel controlled SMC (cSMC) method that can significantly outperform the APF. 
Central to our approach is to take the entire observation sequence $y_{0:T}$ into account by 
constructing proposal distributions that approximate the smoothing distribution $p_{\theta}(x_{0:T}|y_{0:T})$.
% As an alternative to SMC methods, we consider various MCMC algorithms to sample from the smoothing distribution 
% in Section \ref{sub:sis_gibbs}. 
Using a simulated dataset, 
in Section~\ref{sub:sis_illustrate_smc} we show that cSMC provides orders of magnitude improvements 
over BPF and APF in terms of estimating the marginal likelihood. 
In Section~\ref{sub:sis_parameter_inference}, we illustrate the behaviour of marginal likelihood as the number of 
observations increases, and perform parameter inference and predictions. 
We consider MCMC strategies as alternatives to SMC-based methods in Appendix \ref{sub:mcmc}.

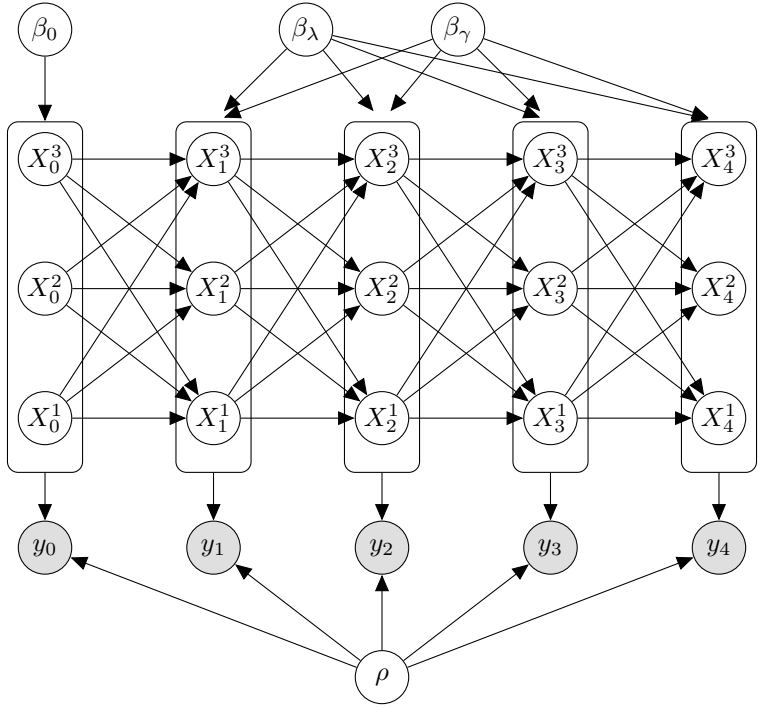
\begin{figure}[htbp]
\centering
\begin{tikzpicture}
% Define nodes
\node[obs](y0) {$y_0$};
\node[latent, above=of y0, ] (x01) {$X_{0}^1$};
\node[latent, above=of x01, ] (x02) {$X_{0}^2$};
\node[latent, above=of x02, ] (x03) {$X_{0}^3$};

\node[obs, right = of y0, xshift = 0.5cm] (y1) {$y_1$} ;
\node[latent, above=of y1, ] (x11) {$X_{1}^1$};
\node[latent, above=of x11, ]  (x12) {$X_{1}^2$};
\node[latent, above=of x12, ] (x13) {$X_{1}^3$};

\node[obs, right = of y1, xshift = 0.5cm] (y2) {$y_2$} ;
\node[latent, above=of y2, ] (x21) {$X_{2}^1$};
\node[latent, above=of x21, ]  (x22) {$X_{2}^2$};
\node[latent, above=of x22, ] (x23) {$X_{2}^3$};

\node[obs, right = of y2, xshift = 0.5cm] (y3) {$y_3$} ;
\node[latent, above=of y3, ] (x31) {$X_{3}^1$};
\node[latent, above=of x31, ]  (x32) {$X_{3}^2$};
\node[latent, above=of x32, ] (x33) {$X_{3}^3$};

\node[obs, right = of y3, xshift = 0.5cm] (y4) {$y_4$} ;
\node[latent, above=of y4, ] (x41) {$X_{4}^1$};
\node[latent, above=of x41, ]  (x42) {$X_{4}^2$};
\node[latent, above=of x42, ] (x43) {$X_{4}^3$};

\node[latent, below = 1cm of y2]            (rho) {$\rho$};
\node[latent, above =of x23, xshift = 1cm] (beta){$\beta_{\gamma}$};
\node[latent, above =of x23, xshift = -1cm](gamma) {$\beta_{\lambda}$};

% Connect the nodes
\edge {x01,x02,x03} {x11} ; %
\edge {x01,x02,x03} {x12} ; %
\edge {x01,x02,x03} {x13} ; %

\edge {x11,x12,x13} {x21} ; %
\edge {x11,x12,x13} {x22} ; %
\edge {x11,x12,x13} {x23} ; %

\edge {x21,x22,x23} {x31} ; %
\edge {x21,x22,x23} {x32} ; %
\edge {x21,x22,x23} {x33} ; %

\edge {x31,x32,x33} {x41} ; %
\edge {x31,x32,x33} {x42} ; %
\edge {x31,x32,x33} {x43} ; %

% Plates
\plate {x0} {(x01)(x02)(x03)} {} ;
\plate {x1} {(x11)(x12)(x13)} {} ;
\plate {x2} {(x21)(x22)(x23)} {} ;
\plate {x3} {(x31)(x32)(x33)} {} ;
\plate {x4} {(x41)(x42)(x43)} {} ;

% connect plates and observations
\edge{x0}{y0};
\edge{x1}{y1};
\edge{x2}{y2};
\edge{x3}{y3};
\edge{x4}{y4};

% connect rho to the observations
\edge{rho}{y0};
\edge{rho}{y1};
\edge{rho}{y2};
\edge{rho}{y3};
\edge{rho}{y4};

\node[above = 0.005cm of x33](x3p){};
\node[above = 0.005cm of x43](x4p){};
\node[above = 0.005cm of x23](x2p){};
\node[above = 0.005cm of x13](x1p){};
% connect gamma to hidden states 
\edge{gamma}{x1p};
\edge{gamma}{x2p};
\edge{gamma}{x3p};
\edge{gamma}{x4p};

% connect gamma to hidden states 
\edge{beta}{x1p};
\edge{beta}{x2p};
\edge{beta}{x3p};
\edge{beta}{x4p};

% connect p0 to hidden states
\node[latent, above = of x03](p0) {$\beta_0$};
\edge{p0}{x0};
\end{tikzpicture}
\caption{Graphical model representation of the agent-based SIS model in Section~\ref{sec:sis} 
for $T = 4$ time steps and a fully connected network with $N = 3$ agents.}
\label{fig:graphical_model}
\end{figure}

\subsection{Sequential Monte Carlo}
\label{sub:sis_apf}
SMC methods \citep{liu1998sequential,doucet2001introduction}, also known as particle filters, 
provide approximations of $p_{\theta}(y_{0:T})$ and $p_{\theta}(x_{0:T}|y_{0:T})$ 
by simulating an interacting particle system of size $P\in\mathbb{N}$. 
In the following, we give a generic description of SMC to include several algorithms in a common framework. 

At the initial time, one samples $P$ configurations of the population from 
a proposal distribution $q_0$ on $\{0,1\}^N$, i.e. 
$X_0^{(p)}=(X_0^{(p),n})_{n\in[1:N]}\sim q_0(\cdot|\theta)$ independently for $p\in[1:P]$. 
Each possible configuration is then assigned a weight 
$W_0^{(p)}\propto w_0(X_0^{(p)})$ that is normalized to sum to one. 
To focus our computation on the more likely configurations, we perform an operation known as resampling 
that discards some configurations and duplicates others according to their weights. 
Each resampling scheme involves sampling ancestor indexes 
$(A_0^{(p)} )_{p\in[1:P]}\in[1:P]^P$ from a distribution $r(\cdot|W_0^{(1:P)})$
on $[1:P]^P$. 
The simplest scheme is multinomial resampling \citep{gordon1993novel}, which samples 
$(A_0^{(p)} )_{p\in[1:P]}$ independently from the categorical distribution on $[1:P]$ 
with probabilities $W_0^{(1:P)}$; other lower variance and adaptive resampling schemes can also be 
employed \citep{fearnhead2003line,gerber2019negative}. Subsequently, for time step $t\in[1:T]$, one 
propagates each resampled configuration according to a proposal transition $q_t$ on $\{0,1\}^N$, i.e. 
$X_t^{(p)}=(X_t^{(p),n})_{n\in[1:N]}\sim q_t(\cdot|X_{t-1}^{(A_{t-1}^{(p)})},\theta)$ independently for $p\in[1:P]$. 
As before, we weight each new configuration according to $W_t^{(p)}\propto w_t(X_{t-1}^{(A_{t-1}^{(p)})}, X_t^{(p)})$, 
and for $t<T$, resample by drawing the ancestor indexes $(A_t^{(p)} )_{p\in[1:P]}\sim r(\cdot|W_t^{(1:P)})$. 
For notational simplicity, we suppress notational dependence of the weight functions $(w_t)_{t\in[0:T]}$ on 
the parameter $\theta$. 
To approximate the desired quantities $p_{\theta}(y_{0:T})$ and $p_{\theta}(x_{0:T}|y_{0:T})$, 
these weight functions have to satisfy 
\begin{align}\label{eqn:weights_requirement}
	w_0(x_0)\prod_{t=1}^Tw_t(x_{t-1}, x_t) = \frac{\mu_{\theta}(x_0)\prod_{t=1}^T
	f_{\theta}(x_t|x_{t-1})\prod_{t=0}^Tg_{\theta}(y_t|x_t)}{q_0(x_0|\theta)\prod_{t=1}^Tq_t(x_t|x_{t-1},\theta)}.
\end{align}

Given the output of the above simulation, an unbiased estimator of the marginal likelihood $p_{\theta}(y_{0:T})$ is 
\begin{align}\label{eqn:smc_likelihood}
	\hat{p}_{\theta}(y_{0:T}) = \left\lbrace\frac{1}{P}\sum_{p=1}^Pw_0(X_0^{(p)})\right\rbrace
	\left\lbrace\prod_{t=1}^T\frac{1}{P}\sum_{p=1}^Pw_t(X_{t-1}^{(A_{t-1}^{(p)})}, X_t^{(p)})\right\rbrace,
\end{align}
and a particle approximation of the smoothing distribution 
is given by 
\begin{align}\label{eqn:smc_smoothing_approx}
	\hat{p}_{\theta}(x_{0:T}|y_{0:T}) = \sum_{p=1}^PW_T^{(p)}\delta_{X_{0:T}^{(p)}}(x_{0:T}).
\end{align}
In the latter approximation, 
each trajectory $X_{0:T}^{(p)}$ is formed by tracing the ancestral lineage of $X_T^{(p)}$, i.e. 
$X_{0:T}^{(p)}=(X_t^{(B_t^{(p)})})_{t\in[0:T]}$ with 
$B_T^{(p)}=p$ 
and $B_t^{(p)} = A_t^{(B_{t+1}^{(p)})}$ for $t\in[0:T-1]$. 
Convergence properties of these approximations as the size of the particle system $P\rightarrow\infty$ 
are well-studied; see for example \citet{del2004feynman}. 
However, the quality of these approximations depends crucially on the choice of proposals $(q_t)_{t\in[0:T]}$ and the 
corresponding weight functions $(w_t)_{t\in[0:T]}$ that satisfy \eqref{eqn:weights_requirement}. 

The BPF of \citet{gordon1993novel} can be recovered by employing the proposals 
$q_0(x_0|\theta)=\mu_{\theta}(x_0), q_t(x_t|x_{t-1},\theta)=f_{\theta}(x_t|x_{t-1})$ for $t\in[1:T]$ and the weight functions 
$w_t(x_t)=g_{\theta}(y_t|x_t)$ for $t\in[0:T]$. 
Although the BPF only costs $\bigo(NTP)$ 
to implement and has convergence guarantees as $P\rightarrow\infty$, 
the variance of its marginal likelihood estimator can be too large to deploy within particle MCMC schemes 
for practical values of $P$ (see Section \ref{sub:sis_illustrate_smc}). Another issue with 
its marginal likelihood estimator, for this particular choice of observation equation,
is that it can collapse to zero if all proposed configurations have less
infections than the observed value, i.e. there exists $t\in[0:T]$ such that $I(X_t^{(p)})<y_t$ for all $p\in[1:P]$. 
With increased cost, this issue can be circumvented using the alive particle filter of \citet{moral2015alive},  
by repeatedly drawing samples at each time step until there are $P$ configurations 
with infections that are larger than or equal to the observed value. 

Alternatively, one can construct better proposals with supports that respect these observational constraints. 
One such option is the fully adapted APF \citep{pitt1999filtering,carpenter1999improved} that corresponds to having the proposals 
$q_0(x_0|\theta)=p_{\theta}(x_0|y_0), q_t(x_t|x_{t-1},\theta)=p_{\theta}(x_t|x_{t-1},y_t)$ for $t\in[1:T]$ and the weight functions 
$w_0(x_0)=p_{\theta}(y_0)$ and $w_t(x_{t-1})=p_{\theta}(y_t|x_{t-1})$ for $t\in[1:T]$. 
At the initial time, computing the marginal likelihood $p_{\theta}(y_0)$ and sampling from the posterior of agent states
$p_{\theta}(x_0|y_0)$ can be done exactly (or approximately) as described in Sections 
\ref{sub:static_likelihood} and \ref{sub:static_posterior_states} respectively for the static model. 
More precisely, we compute 
\begin{align}\label{eqn:sis_initial_likelihood}
	p_{\theta}(y_0) = \sum_{i_0=0}^N \PB(i_0;\alpha_0) \Binomial(y_0; i_0, \rho) \mathbbm{1}(i_0\geq y_0)
\end{align}
and sample from 
\begin{align}\label{eqn:sis_apf_initial}
	p_{\theta}(x_0,i_0|y_0) = p_{\theta}(i_0|y_0)p_{\theta}(x_0|i_0) = \frac{\PB(i_0;\alpha_0) \Binomial(y_0; i_0, \rho) 
	\mathbbm{1}(i_0 \geq y_0)}{p_{\theta}(y_0)}\CB(x_0;\alpha_0, i_0), 
\end{align}
which admits $p_{\theta}(x_0|y_0)$ as its marginal distribution. 
For time step $t\in[1:T]$, by conditioning on the previous configuration $x_{t-1}\in\{0,1\}^N$, 
the same ideas can be used to compute the predictive likelihood $p_{\theta}(y_t|x_{t-1})$ and sample from 
the transition $p_{\theta}(x_t|x_{t-1},y_t)$, i.e. we compute 
\begin{align}
	p_{\theta}(y_t|x_{t-1}) = \sum_{i_{t}=0}^N \PB(i_{t};\alpha(x_{t-1})) \Binomial(y_t; i_{t}, \rho) \mathbbm{1}(i_{t}\geq y_t)
\end{align}
and sample from 
\begin{align}\label{eqn:sis_apf_transition}
	p_{\theta}(x_t,i_t|x_{t-1},y_t) &= p_{\theta}(i_t|x_{t-1},y_t)p_{\theta}(x_t|x_{t-1},i_t) \\
	& = \frac{\PB(i_t;\alpha(x_{t-1})) \Binomial(y_t; i_t, \rho) 
	\mathbbm{1}(i_t \geq y_t)}{p_{\theta}(y_t|x_{t-1})}\CB(x_t;\alpha(x_{t-1}), i_t)\notag
\end{align}
which admits $p_{\theta}(x_t,|x_{t-1},y_t)$ as its marginal transition. 

An algorithmic description of the resulting APF is detailed in Algorithm
\ref{algo:apf_sis}, where the notation $\Cat([0:N],V^{(0:N)})$ refers to the
categorical distribution on $[0:N]$ with probabilities
$V^{(0:N)}=(V^{(i)})_{i\in[0:N]}$.  As the weights in the fully adapted APF at
time $t\in[1:T]$ only depend on the configuration at time $t-1$, note that we
have interchanged the order of sampling and resampling to promote sample
diversity. The cost of running APF exactly is $\bigo(N^2TP)$. 
To reduce the computational cost to $\bigo(N\log(N) T P)$ 
one can approximate the above Poisson binomial PMFs with the translated Poisson
approximation \eqref{eqn:transpoi}, and employ MCMC to sample from the above 
conditioned Bernoulli distributions. 

\begin{algorithm}
	\SetAlgoLined
	\KwIn{Parameters $\theta\in\Theta$ and number of particles $P\in\mathbb{N}$} 
	compute $v_0^{(i)} = \PB(i;\alpha_0) \Binomial(y_0; i, \rho) \mathbbm{1}(i\geq y_0)$ for $i\in[0:N]$\\ 
	set $w_0 = \sum_{i=0}^Nv_0^{(i)}$\\ 
	normalize $V_0^{(i)} = v_0^{(i)} / w_0$ for $i\in[0:N]$\\ 	
	sample $I_0^{(p)}\sim\Cat([0:N],V_0^{(0:N)})$ and $X_0^{(p)} | I_0^{(p)} \sim\CB(\alpha_0, I_0^{(p)})$ for $p\in[1:P]$\\
	\For {$t = 1,\cdots, T$ and $p = 1,\cdots, P$}{
		compute $v_t^{(i,p)} = \PB(i;\alpha(X_{t-1}^{(p)})) \Binomial(y_t; i, \rho) \mathbbm{1}(i\geq y_t)$ for $i\in[0:N]$\\ 		
		set $w_t^{(p)} = \sum_{i=0}^Nv_t^{(i,p)}$\\ 			
		normalize $V_t^{(i,p)} = v_t^{(i,p)} / w_t^{(p)}$ for $i\in[0:N]$\\ 	
		normalize $W_{t}^{(p)} = w_t^{(p)} / \sum_{k=1}^Pw_t^{(k)}$\\
		sample $A_{t-1}^{(p)}\sim r(\cdot|W_{t}^{(1:P)})$\\
		sample $I_t^{(p)}\sim\Cat([0:N],V_t^{(0:N,A_{t-1}^{(p)})})$ and $X_t^{(p)} | I_t^{(p)} \sim\CB(\alpha(X_{t-1}^{(A_{t-1}^{(p)})}), I_t^{(p)})$\\
	}
	
	\KwOut{Marginal likelihood estimator $\hat{p}_{\theta}(y_{0:T})=w_0\prod_{t=1}^TP^{-1}\sum_{p=1}^Pw_t^{(p)}$, states 
	$(X_t^{(p)})_{(t,p)\in[0:T]\times[1:P]}$ and ancestors $(A_t^{(p)})_{(t,p)\in[0:T-1]\times[1:P]}$}
	\caption{Auxiliary particle filter for SIS model}\label{algo:apf_sis}
\end{algorithm}

\subsection{Controlled sequential Monte Carlo}
\label{sub:sis_controlled_smc}

To obtain better performance than APF, we can construct proposals that take not just the next but all future observations 
into account. 
We can sequentially decompose the smoothing distribution as
\begin{align}\label{eqn:smoothing_distribution_bif}
	p_{\theta}(x_{0:T}|y_{0:T}) = p_{\theta}(x_0|y_{0:T})\prod_{t=1}^Tp_{\theta}(x_t|x_{t-1},y_{t:T}),
\end{align}
and it follows that the optimal proposal is 
$q_0^\star(x_0|\theta)=p_{\theta}(x_0|y_{0:T})$ and 
$q_t^\star(x_t|x_{t-1},\theta)=p_{\theta}(x_t|x_{t-1},y_{t:T})$ 
for $t\in[1:T]$, as this gives exact samples from the smoothing 
distribution.
The resulting SMC marginal likelihood estimator in 
\eqref{eqn:smc_likelihood} would have zero variance for any choice of weight functions satisfying \eqref{eqn:weights_requirement}. 
To design approximations of the optimal proposal, it will be instructive to rewrite it as 
\begin{align}\label{eqn:decompose_smoothing}
	p_{\theta}(x_0|y_{0:T}) = \frac{\mu_{\theta}(x_0)\psi_0^\star(x_0)}{\mu_{\theta}(\psi_0^\star)}, 
	\quad p_{\theta}(x_t|x_{t-1},y_{t:T}) = \frac{f_{\theta}(x_t|x_{t-1})\psi_t^\star(x_t)}{f_{\theta}(\psi_t^\star|x_{t-1})}, \quad t\in[1:T],
\end{align}
where $\psi_t^\star(x_t)=p(y_{t:T}|x_t)$ for $t\in[0:T]$, 
$\mu_{\theta}(\psi_0^\star) = \sum_{x_0\in\{0,1\}^N}\mu_{\theta}(x_0)\psi_0^\star(x_0)$ denotes the expectation of $\psi_0^\star$ 
with respect to $\mu_{\theta}$ and 
$f_{\theta}(\psi_t^\star|x_{t-1})=\sum_{x_t\in\{0,1\}^N}f_{\theta}(x_t|x_{t-1})\psi_t^\star(x_t)$ denotes the conditional expectation 
of $\psi_t^\star$ under $f_{\theta}$.
Equation \eqref{eqn:decompose_smoothing} shows how the latent process \eqref{eqn:dynamic_X}, defined by $\mu_{\theta}$ 
and $f_{\theta}$, should be modified to obtain the optimal proposal. The functions $(\psi_t^\star)_{t\in[0:T]}$, known as 
the backward information filter (BIF) \citep{bresler1986two,briers2010smoothing}, can be defined using the backward recursion 
\begin{align}\label{eqn:backward_info_filter}
	\psi_T^\star(x_T)=g_{\theta}(y_T|x_T), \quad \psi_t^\star(x_t) = g_{\theta}(y_t|x_t)f_{\theta}(\psi_{t+1}^\star|x_t),\quad t\in[0:T-1],
\end{align}
which shows how information from future observations are propagated backwards over time. 
As the cost of computing and storing the BIF using the recursion \eqref{eqn:backward_info_filter} are 
$\bigo(2^{2N}T)$ and $\bigo(2^NT)$, respectively, approximations are necessary when $N$ is large. 
In contrast to \citet{guarniero2017iterated,heng2020controlled} that rely on regression techniques to approximate 
the BIF, our approach is based on dimensionality reduction by coarse-graining the agent-based model. 

At the terminal time $T$, the function $\psi_T^\star(x_T)=\Binomial(y_T;I(x_T), \rho)\mathbbm{1}(I(x_T) \geq y_T)$
only depends on the agent states $x_T\in\{0,1\}^N$ through the one-dimensional summary $I(x_T)\in[0:N]$. 
Therefore it suffices to compute and store $\psi_T(i_T)=\Binomial(y_T;i_T, \rho)\mathbbm{1}(i_T \geq y_T)$ for all $i_T\in[0:N]$. 
Note that $\psi_T(I(x_T))$ should be seen as a function of the agent states. 
To iterate the backward recursion \eqref{eqn:backward_info_filter}, 
we have to compute the conditional expectation $f_{\theta}(\psi_{T}|x_{T-1})= \sum_{i_T=0}^N \PB(i_T;\alpha(x_{T-1})) \psi_T(i_T)$
for all $x_{T-1}\in\{0,1\}^N$. By a thinning argument (Appendix \ref{sub:thinning}), this is equal to 
$\PB(y_T;\rho \, \alpha(x_{T-1}))$. Hence it is clear that all $2^N$ possible configurations of the population 
have to be considered to iterate the recursion, and an approximation of $\alpha(x_{T-1})$ is necessary at this point. 

To reduce dimension, we consider 
\begin{equation}
	\label{eqn:sis_alpha_approx}
		\bar{\alpha}^n(x_{T-1}) = \begin{cases}
					\bar{\lambda} N^{-1}I(x_{T-1}),\quad & \mbox{if }x_{T-1}^n=0,\\
					1-\bar{\gamma},\quad & \mbox{if }x_{T-1}^n=1.
					\end{cases}
\end{equation}
This amounts to approximating the proportion of infected neighbours of each agent by the population proportion of infections,
and replacing individual infection and recovery rates in~\eqref{eqn:sis_alpha_hetero} by their population averages, i.e. 
$\lambda^n\approx\bar{\lambda}=N^{-1}\sum_{n=1}^N\lambda^n$ and  
$\gamma^n\approx\bar{\gamma}=N^{-1}\sum_{n=1}^N\gamma^n$. 
Writing $\bar{f}_{\theta}(x_T|x_{T-1})=\prod_{n=1}^N\Bernoulli(x_{T}^n;\bar{\alpha}^n(x_{T -1}))$ as  
the Markov transition associated to \eqref{eqn:sis_alpha_approx}, 
we approximate the conditional expectation $f_{\theta}(\psi_{T}|x_{T-1})$ by
\begin{align}\label{eqn:sis_condexp_approx}
	\bar{f}_{\theta}(\psi_{T}|I(x_{T-1})) = \sum_{i_T=0}^N 
	\SB(i_T; N-I(x_{T-1}), \bar{\lambda} N^{-1}I(x_{T-1}), I(x_{T-1}), 1-\bar{\gamma})\psi_T(i_T),
\end{align}
where $\SB(N_1,p_1,N_2,p_2)$ denotes the distribution of a sum of two independent 
$\Binomial(N_1,p_1)$ and $\Binomial(N_2,p_2)$ random variables. 
This follows as a Poisson binomial distribution with homogeneous probabilities \eqref{eqn:sis_alpha_approx} reduces to 
the $\SB$ distribution in \eqref{eqn:sis_condexp_approx}, which is not analytically tractable 
but can be computed exactly in $\bigo(N^2)$ cost using a naive implementation of the convolution\footnote{$\SB(i;N_1,p_1,N_2,p_2) = 
\sum_{j=0}^i \Binomial(j;N_1,p_1)\Binomial(i-j;N_2,p_2)$ for $i\in[0:N_1+N_2]$.}. 
% Although the fast Fourier transform can reduce the complexity to $\bigo(N \log (N))$, we found that this 
% can be numerically unstable due to underflow issues when $N$ is reasonably large. 
In the large $N$ regime, we advocate approximating  
the $\SB$ distribution with the translated Poisson \eqref{eqn:transpoi}, which reduces the cost to $\bigo(N)$ 
at the price of an approximation error. Moreover, the latter bias only affects the quality of our proposal distributions, 
and not the consistency properties of SMC approximations. 
The resulting approximation of $\psi_{T-1}^\star(x_{T-1})$ can be computed using 
\begin{align}
	\psi_{T-1}(i_{T-1})=\Binomial(y_{T-1};i_{T-1}, \rho)\mathbbm{1}(i_{T-1} \geq y_{T-1})\bar{f}_{\theta}(\psi_{T}|i_{T-1})
\end{align}
for all $i_{T-1}\in[0:N]$. This costs $\bigo(N^3)$ if convolutions are implemented naively and $\bigo(N^2)$ if 
translated Poisson approximations of \eqref{eqn:sis_condexp_approx} are employed. 
We then continue in the same manner to approximate the recursion 
\eqref{eqn:backward_info_filter} until the initial time. Algorithm \ref{algo:bif_sis} summarizes 
our approximation of the BIF $(\psi_t)_{t\in[0:T]}$, which costs $\bigo(N^3T)$ or $\bigo(N^2T)$ to compute and $\bigo(NT)$ in storage. 

Our corresponding approximation of the optimal proposal \eqref{eqn:decompose_smoothing} is
\begin{align}\label{eqn:csmc_proposals}
	q_0(x_0|\theta) = \frac{\mu_{\theta}(x_0)\psi_0(I(x_0))}{\mu_{\theta}(\psi_0)}, 
	\quad q_t(x_t|x_{t-1},\theta) = \frac{f_{\theta}(x_t|x_{t-1})\psi_t(I(x_t))}{f_{\theta}(\psi_t|x_{t-1})}, \quad t\in[1:T].
\end{align}
To employ these proposals within SMC, the appropriate weight functions \citep{scharth2016particle} satisfying \eqref{eqn:weights_requirement} are 
\begin{align}\label{eqn:csmc_weights}
	&w_0(x_0) = \frac{\mu_{\theta}(\psi_0)g_{\theta}(y_0|x_0)f_{\theta}(\psi_1|x_0)}{\psi_0(I(x_0))},
	&w_t(x_t) = \frac{g_{\theta}(y_t|x_t)f_{\theta}(\psi_{t+1}|x_t)}{\psi_t(I(x_t))},\quad t\in[1:T-1],
\end{align}
and $w_T(x_T) = 1$. To evaluate the weights, note that expectations can be computed as
\begin{align}
	\mu_{\theta}(\psi_0) = \sum_{i_0=0}^N\PB(i_0;\alpha_0)\psi_0(i_0),\quad 
	f_{\theta}(\psi_{t}|x_{t-1}) = \sum_{i_t=0}^N\PB(i_t;\alpha(x_{t-1}))\psi_t(i_t),\quad t\in[1:T].
\end{align}
Sampling from the proposals in \eqref{eqn:csmc_proposals} can be performed in a similar manner as the APF. 
To initialize, we sample from 
\begin{align}
	q_0(x_0,i_0|\theta) = q_0(i_0|\theta)q_0(x_0|i_0,\theta) = \frac{\PB(i_0;\alpha_0)\psi_0(i_0)}{\mu_{\theta}(\psi_0)}
	\CB(x_0;\alpha_0, i_0) 
\end{align}
which admits $q_0(x_0|\theta)$ as its marginal distribution, and for time $t\in[1:T]$
\begin{align}
	q_t(x_t,i_t|x_{t-1},\theta) & = q_t(i_t|x_{t-1},\theta)q_t(x_t|x_{t-1},i_t,\theta)\\
	& = \frac{\PB(i_t;\alpha(x_{t-1}))\psi_t(i_t)}{f_{\theta}(\psi_{t}|x_{t-1})}\CB(x_t;\alpha(x_{t-1}), i_t)\notag
\end{align}
which admits $q_t(x_t|x_{t-1},\theta)$ as its marginal transition. 
Algorithm \ref{algo:csmc_sis} gives an algorithmic summary of the resulting SMC method, 
which we shall refer to as controlled SMC (cSMC), following the terminology of \citet{heng2020controlled}. 
This has the same cost as the APF, and one can also employ translated Poisson approximations \eqref{eqn:transpoi} and 
MCMC to reduce the computational cost.

To study the performance of cSMC, we consider the Kullback--Leibler (KL) divergence 
from our proposal distribution $q_{\theta}(x_{0:T})=q_0(x_0|\theta)\prod_{t=1}^Tq_t(x_t|x_{t-1},\theta)$ to 
the smoothing distribution $p_{\theta}(x_{0:T}|y_{0:T})$, denoted as $\KL\left(p_{\theta}(x_{0:T}|y_{0:T}) \mid q_{\theta}(x_{0:T})\right)$, 
which characterizes the quality of our importance proposal \citep{chatterjee2018sample}. 
The following result provides a decomposition of this KL divergence in terms of 
logarithmic differences between the BIF $(\psi_t^\star)_{t\in[0:T]}$ and our approximation $(\psi_t)_{t\in[0:T]}$ 
under the marginal distributions of the smoothing distribution and our proposal distribution, 
denoted as $\eta_t^\star(x_t|\theta)$ and $\eta_t(x_t|\theta)$ respectively for each time $t\in[0:T]$. 
Given a function $\varphi:\{0,1\}^N\rightarrow\mathbb{R}$, we will write 
$\eta_t^\star(\varphi|\theta)$ and $\eta_t(\varphi|\theta)$ to denote expectations 
under these marginal distributions, and its corresponding $L^2$-norms as 
$\|\varphi\|_{L^2(\eta_t^\star)}=\eta_t^\star(\varphi^2|\theta)^{1/2}$ and 
$\|\varphi\|_{L^2(\eta_t)}=\eta_t(\varphi^2|\theta)^{1/2}$.

\begin{proposition}\label{prop:KL_BIF_bound}
The Kullback--Leibler divergence from $q_{\theta}(x_{0:T})$ to $p_{\theta}(x_{0:T}|y_{0:T})$ satisfies 
\begin{align}\label{eqn:KL_BIF_bound}
	\KL\left(p_{\theta}(x_{0:T}|y_{0:T}) \mid q_{\theta}(x_{0:T})\right) 
	\leq \sum_{t=0}^T \eta_{t}^{\star}(\log(\psi_t^{\star}/\psi_t)|\theta) + M_{t-1} \eta_t(\log(\psi_t/\psi_t^{\star})|\theta),
\end{align}
where $M_{-1}=1$ and $M_t=\max_{x_{t}\in\{0,1\}^N}\eta_{t}^\star(x_{t}|\theta)/\eta_{t}(x_{t}|\theta)$ 
for $t\in[0:T-1]$.
\end{proposition}

The proof is given in Appendix \ref{sec:KL_proposal}. We show in Appendix \ref{sec:ratio_marginal} 
that the constants $(M_t)_{t\in[0:T-1]}$ are finite by 
upper bounding the weight functions in \eqref{eqn:csmc_weights}. 
The next result characterizes the error of our BIF approximation measured in terms of  
the KL upper bound in \eqref{eqn:KL_BIF_bound}. 

\begin{proposition}\label{prop:BIF_error}
For each time $t\in[0:T]$, the BIF approximation in Algorithm \ref{algo:bif_sis} satisfies:
\begin{align}
\eta_{t}^{\star}(\log(\psi_t^{\star}/\psi_t)|\theta) \leq \sum_{k=t}^{T-1}c_{\theta}^\star(\psi_{k+1})
\|\phi_{\theta}^\star \|_{L^2(\eta_k^\star)}\|\Delta_{\theta}\|_{L^2(\eta_k^\star)}, \label{eqn:BIF_logerror1}\\
\eta_t(\log(\psi_t/\psi_t^{\star})|\theta) \leq \sum_{k=t}^{T-1}c_{\theta}(\psi_{k+1})
\|\phi_{\theta} \|_{L^2(\eta_k)}\|\Delta_{\theta}\|_{L^2(\eta_k)},\label{eqn:BIF_logerror2}
\end{align}
with constants
\begin{align*}
& c_{\theta}^\star(\psi_{k})=\frac{1}{\min_{x_{k-1}\in\{0,1\}^N\setminus (0,\ldots,0)}f_{\theta}(\psi_{k}|x_{k-1})}<\infty, 
& c_{\theta}(\psi_{k})=\frac{1}{\min_{i_{k-1}\in[1:N]}\bar{f}_{\theta}(\psi_{k}|i_{k-1})\}}<\infty
\end{align*}
for $k \in [1 : T]$, functions $\Delta_{\theta}(x)=\sum_{n=1}^N|\bar{\alpha}^n(x)-\alpha^n(x)|$, 
\begin{align}
&\phi_{\theta}^\star(x) = \left\{\sum_{i\in[0:N]} \frac{\PB(i;\alpha(x))^2}{\PB(i;\bar{\alpha}(x))^2}\right\}^{1/2}, 
&\phi_{\theta}(x) = \left\{\sum_{i\in[0:N]} \frac{\PB(i;\bar{\alpha}(x))^2}{\PB(i;\alpha(x))^2}\right\}^{1/2}
\end{align}
for $x \in \{0,1\}^N$, and the convention that $\sum_{k=t}^p \{\cdot\}=0$ for $p<t$.
\end{proposition}

The proof of Proposition \ref{prop:BIF_error} in Appendix \ref{sec:bif_approx}
derives recursive bounds of the approximation errors.   
The crux of our arguments is to upper bound the error of taking conditional expectations 
under the homogeneous probabilities \eqref{eqn:sis_alpha_approx}. This relies on an upper bound of the 
Kullback--Leibler divergence between two Poisson binomial distributions 
that is established in Appendix \ref{sec:compare_poibin}, which may be of independent interest. 
If conditional expectations are further approximated using the translated Poisson approximations, 
one can also employ the results by \citet{cekanavicius2001centered,barbour2002total} to incorporate these errors in our analysis.
The error bounds in \eqref{eqn:BIF_logerror1} and \eqref{eqn:BIF_logerror2} show how the accuracy 
of the BIF approximation depend on our approximation of the success probability 
$\alpha(x)$ via the term $\Delta_{\theta}(x)$. If we decompose 
$\Delta_{\theta}(x)\leq \Delta_{\theta}^{\mathcal{G}}(x) + \Delta_{\theta}^{\lambda} + \Delta_{\theta}^{\gamma}$, where
\begin{align}\label{eqn:Delta_decomposition}
	\Delta_{\theta}^{\mathcal{G}}(x) = \sum_{n=1}^{N}\left| N^{-1}I(x)-\mathcal{D}(n)^{-1}\sum_{m\in\mathcal{N}(n)}x^{m}\right|,\quad
	\Delta_{\theta}^{\lambda} = \sum_{n=1}^{N}|\bar{\lambda}-\lambda^{n}|, \quad
	\Delta_{\theta}^{\gamma}=\sum_{n=1}^{N}|\bar{\gamma}-\gamma^{n}|,
\end{align}
we see the effect of coarse-graining the agent-based model in \eqref{eqn:sis_alpha_approx}. 
In Appendix \ref{sec:finer_BIF_approx}, we discuss how to reduce the errors $\Delta_{\theta}^{\lambda}(x)$ and $\Delta_{\theta}^{\gamma}(x)$.  
By adopting a more fine-grained approximation based on clustering of the infection and recovery rates, 
one can obtain more accurate approximations at the expense of increased computational cost. 

\begin{algorithm}
	\SetAlgoLined
	\KwIn{Parameters $\theta\in\Theta$} 
	compute average infection and recovery rates $\bar{\lambda}=N^{-1}\sum_{n=1}^N\lambda^n$, 
	$\bar{\gamma}=N^{-1}\sum_{n=1}^N\gamma^n$ \\
	compute $\psi_T(i_T)=\Binomial(y_T;i_T, \rho)\mathbbm{1}(i_T \geq y_T)$ for $i_T\in[0:N]$\\
	\For {$t = T-1,\cdots, 0$ and $i_t = 0,\cdots, N$}{
	compute $	\bar{f}_{\theta}(\psi_{t+1}|i_{t}) = \sum_{i_{t+1}=0}^N
	\SB(i_{t+1}; N-i_t, \bar{\lambda} N^{-1}i_t, i_t, 1-\bar{\gamma})\psi_{t+1}(i_{t+1})$\\
	compute $	\psi_{t}(i_{t})=\Binomial(y_{t};i_{t}, \rho)\mathbbm{1}(i_{t} \geq y_{t})\bar{f}_{\theta}(\psi_{t+1}|i_{t})$
	}	
	
	\KwOut{Approximate BIF $(\psi_{t})_{t\in[0:T]}$}
	\caption{Backward information filter approximation for SIS model}\label{algo:bif_sis}
\end{algorithm}

\begin{algorithm}
	\SetAlgoLined
	\KwIn{Parameters $\theta\in\Theta$, approximate BIF $(\psi_{t})_{t\in[0:T]}$ 
	and number of particles $P\in\mathbb{N}$} 	
	compute probabilities $v_0^{(i)} = \PB(i;\alpha_0) \psi_0(i)$ for $i\in[0:N]$\\ 
	set $E_0 = \sum_{i=0}^Nv_0^{(i)}$ and normalize probabilities $V_0^{(i)} = v_0^{(i)} / E_0$ for $i\in[0:N]$\\ 	
	sample $I_0^{(p)}\sim\Cat([0:N],V_0^{(0:N)})$ and $X_0^{(p)} | I_0^{(p)} \sim\CB(\alpha_0, I_0^{(p)})$ for $p\in[1:P]$\\
	compute weights $w_0^{(p)} = w_0(X_0^{(p)})$ using \eqref{eqn:csmc_weights} for $p\in[1:P]$\\
	\For {$t = 1,\cdots, T$ and $p = 1,\cdots, P$}{
		normalize weights $W_{t-1}^{(p)} = w_{t-1}^{(p)} / \sum_{k=1}^Pw_{t-1}^{(k)}$\\
		sample ancestors $A_{t-1}^{(p)}\sim r(\cdot|W_{t-1}^{(1:P)})$\\
		compute probabilities $v_t^{(i,p)} = \PB(i;\alpha(X_{t-1}^{(A_{t-1}^{(p)})})) \psi_t(i)$ for $i\in[0:N]$ \\ 	
		set $E_t^{(p)}= \sum_{i=0}^Nv_t^{(i,p)}$ and normalize probabilities $V_t^{(i,p)} = v_t^{(i,p)} / E_t^{(p)}$ for $i\in[0:N]$\\ 	
		sample $I_t^{(p)}\sim\Cat([0:N],V_t^{(0:N,p)})$ and $X_t^{(p)} | I_t^{(p)} \sim\CB(\alpha(X_{t-1}^{(A_{t-1}^{(p)})}), I_t^{(p)})$\\
		compute weights $w_t^{(p)} = w_t(X_t^{(p)})$ using \eqref{eqn:csmc_weights} \\
	}
	
	\KwOut{Marginal likelihood estimator $\hat{p}_{\theta}(y_{0:T})=\prod_{t=0}^TP^{-1}\sum_{p=1}^Pw_t^{(p)}$, states 
	$(X_t^{(p)})_{(t,p)\in[0:T]\times[1:P]}$ and ancestors $(A_t^{(p)})_{(t,p)\in[0:T-1]\times[1:P]}$}
	\caption{Controlled sequential Monte Carlo for SIS model}\label{algo:csmc_sis}
\end{algorithm}

\subsection{Numerical illustration of sequential Monte Carlo methods}
\label{sub:sis_illustrate_smc}
We now illustrate the behaviour of the above SMC methods on simulated data. 
We consider a population of $N=100$ agents that are fully connected for $T=90$ time steps. 
The agent covariates $w^n=(w_1^n,w_2^n)$ are taken as $w_1^n=1$ and sampled from 
$w_2^n\sim\mathcal{N}(0,1)$ independently for $n\in[1:N]$. 
Given these covariates, we simulate data from model 
\eqref{eqn:latentprocess_parameters}-\eqref{eqn:sis_alpha_hetero} and \eqref{eqn:static_Y} 
with the parameters $\beta_0=(-\log(N-1),0)$, $\beta_{\lambda}=(-1,2)$, $\beta_{\gamma}=(-1,-1)$ and $\rho=0.8$. 

The top panel of Figure \ref{fig:sis_illustrate_ess} illustrates the performance of these SMC methods 
at the data-generating parameter (DGP), measured in terms of the effective sample size (ESS) criterion
\citep{kong1994sequential}. The ESS at time $t\in[0:T]$, defined in terms of the normalized 
weights $(W_t^{(p)})_{p\in[1:P]}$ as $1/\sum_{p=1}^P(W_t^{(p)})^2$, measures the adequacy of the 
importance sampling approximation at each step. 
We consider two implementations of the controlled SMC in Algorithm \ref{algo:csmc_sis}: 
cSMC1 employs proposal distributions that are defined by the BIF approximation in Algorithm \ref{algo:bif_sis}, 
while cSMC2 relies on translated Poisson approximations of the SumBin distributions in Algorithm \ref{algo:bif_sis} 
to lower the computational cost. This lowers the run-time of the BIF approximation from $0.35$ to $0.12$ second, 
which are insignificant relative to the cost of cSMC for a population size of $N=100$. 
As expected, the APF performs better than the BPF by taking the next observation into account, and  
cSMC does better than APF by incorporating information from the entire observation sequence. 
Furthermore, the faster BIF approximation does not result in a noticeable loss of cSMC performance 
due to the accuracy of the translated Poisson approximations with $N=100$ agents. 
Although the performance of BPF seems adequate in this simulated data setting, the middle and bottom 
panels of Figure \ref{fig:sis_illustrate_ess} reveal that its particle
approximation can collapse whenever 
there are smaller or larger observation counts. In contrast, the performance of
APF and cSMC appear to be more robust to such informative observations.

\begin{figure}[htbp]
	\centering
    \includegraphics[width=1.0\textwidth]{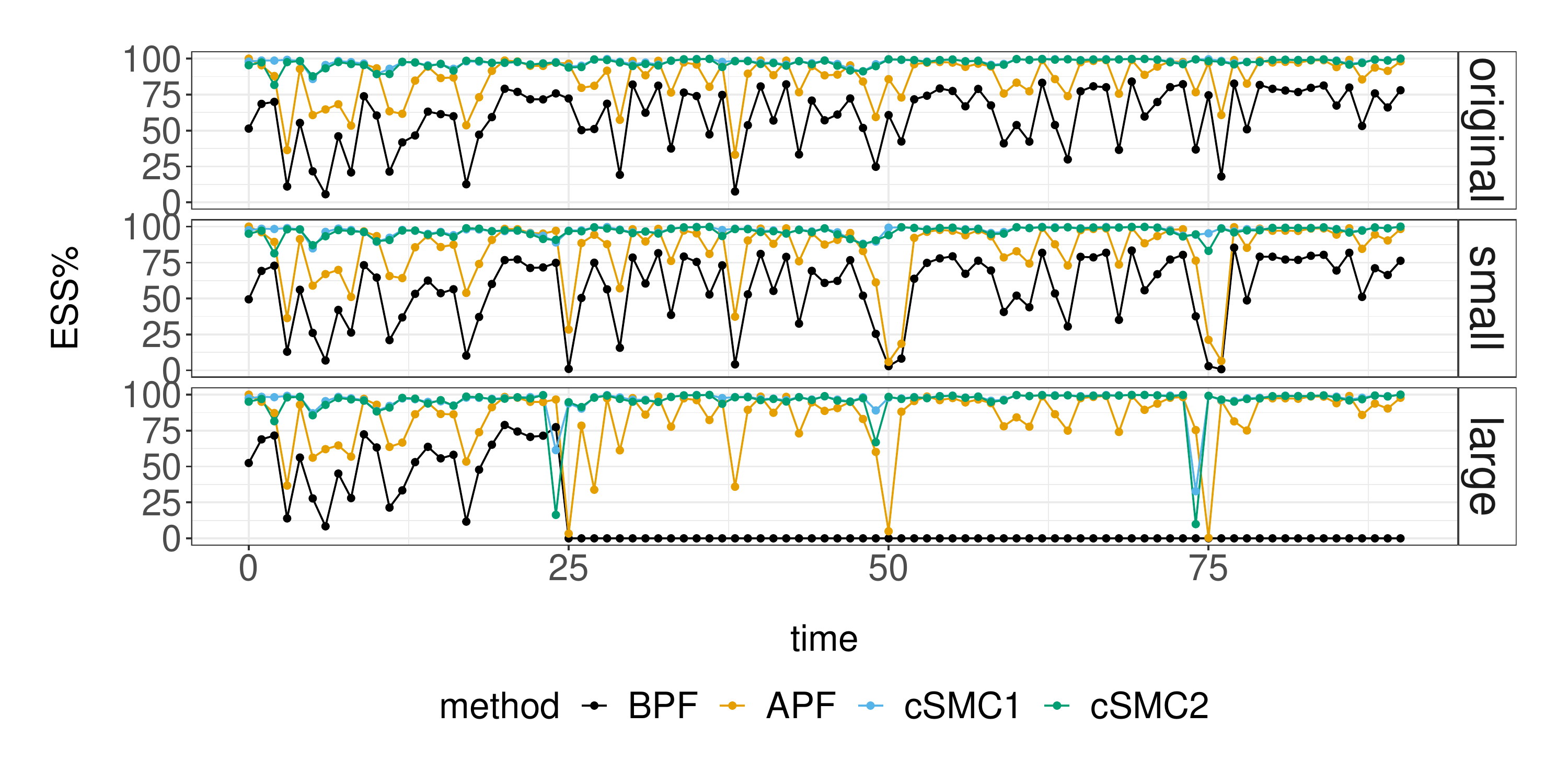}
	\caption{Effective sample size of various SMC methods with $P=512$ particles when considering 
simulated observations (\emph{top}), when observations at $t\in\{25,50,75\}$ are replaced by $\lfloor y_t/2\rfloor$ (\emph{middle}), or $\min(2y_t,N)$ (\emph{bottom}).}
	\label{fig:sis_illustrate_ess}
\end{figure}

Next we examine the performance of these SMC methods in terms of marginal likelihood estimation. 
Table \ref{tab:sis_pf_variance} displays the variance of the log-marginal likelihood estimator at two parameter sets,
and its average cost measured as run-time that were estimated using $100$ independent repetitions of each method. 
At the DGP, it is apparent that the cSMC estimators achieve the asymptotic regime  
of $P\rightarrow\infty$ earlier than BPF and APF, which seem to require at least $P=512$ particles. 
Based on the largest value of $P$ that we considered, the asymptotic variance of APF, cSMC1 and cSMC2 was found to be 
$29$, $155$ and $115$ times smaller relative to BPF, respectively. 
As the cost of APF and cSMC was approximately 
$19$ times more expensive than BPF in our implementation, we see that APF, cSMC1 and cSMC2 are respectively $1.5$, $8$ and $6$ 
times more efficient than BPF at the DGP. We can expect these efficiency gains to be more significant as we move away 
from the DGP. To illustrate this, we consider another parameter set which has $\beta_{\lambda}=(-3,0)$ and keeps 
all other parameters at the DGP. 
Although this is a less likely set of parameters as the log-marginal likelihood is approximately $232$ lower than the DGP, 
adequate marginal likelihood estimation is crucial when employing SMC methods within particle MCMC algorithms for parameter inference. 
In this case, we found that the BPF marginal likelihood estimates could collapse to zero 
for the values of $P$ that are considered in Table \ref{tab:sis_pf_variance}.  
In contrast, APF and cSMC would not suffer from this issue by construction. 
By increasing the number of BPF particles to $P=262,144$ and comparing its performance to 
APF, cSMC1 and cSMC2 with $P=2048$ particles, we find that BPF is respectively $9$, $76$ and $42$ times less efficient 
at this parameter set. Lastly, Figure \ref{fig:sis_viance_rho} illustrates the comparison of SMC methods 
as the parameter $\rho$ varies and all other parameters fixed at the DGP.

% For APF and cSMC, we have an option of bringing 
% the cost down to $\bigo(N)$ using translated Poisson approximations for density evaluations and 
% MCMC for sampling. We also assess the bias introduced by these approximations 
% $\hat{p}(y_{0:T}\mid \theta^{\star}).$

\begin{table}
\centering
\small
\begin{tabular}{c|cc|ccc|ccccc}
\multicolumn{1}{c}{} & \multicolumn{2}{c}{BPF} & \multicolumn{3}{c}{APF} & \multicolumn{2}{c}{cSMC1} & \multicolumn{2}{c}{cSMC2}\\
\hline
$P$ & DGP & Cost & DGP & Non-DGP & Cost & DGP & Non-DGP & DGP & Non-DGP & Cost  \\
 & Var &  (sec) & Var & Var & (sec) & Var & Var & Var & Var & (sec)\\
\hline \hline
64 & 4.32 & 0.09 & 0.281 & 71.88 & 1.49 & 0.0696 & 13.52 & 0.0779 & 18.44 & 1.46\\
128 & 2.39 & 0.17 & 0.154 & 39.52 & 2.95 & 0.0285 & 8.62 & 0.0382 & 8.48 & 2.88\\
256 & 1.67 & 0.33 & 0.110 & 26.86 & 5.85 & 0.0164 & 4.11 & 0.0190 & 6.88 & 5.72\\
512 & 0.88 & 0.63 & 0.056 & 18.98 & 11.72 & 0.0087 & 3.57 & 0.0105 & 5.03 & 11.41\\
1024 & 0.55 & 1.25 & 0.026 & 13.38 & 23.46 & 0.0049 & 2.05 & 0.0046 & 3.41 & 22.83\\
2048 & 0.31 & 2.49 & 0.011 & 9.93 & 47.48 & 0.0020 & 1.15 & 0.0027 & 2.07 & 45.97\\
\end{tabular}
\caption{Variance of log-marginal likelihood estimator and its average cost (in seconds) using different number of particles $P$ 
and various SMC methods. The parameter sets considered are the data generating parameters (DGP) and a modification of the 
DGP with $\beta_{\lambda}=(-3,0)$ (labelled as non-DGP). The cost of cSMC1 and cSMC2 are comparable and the run-times 
to compute the BIF approximations (0.35 and 0.12 second for cSMC1 and cSMC2) are not reported in this table.}
\label{tab:sis_pf_variance}
\end{table}

\begin{figure}[htbp]
\centering
        \includegraphics[width=0.5\textwidth]{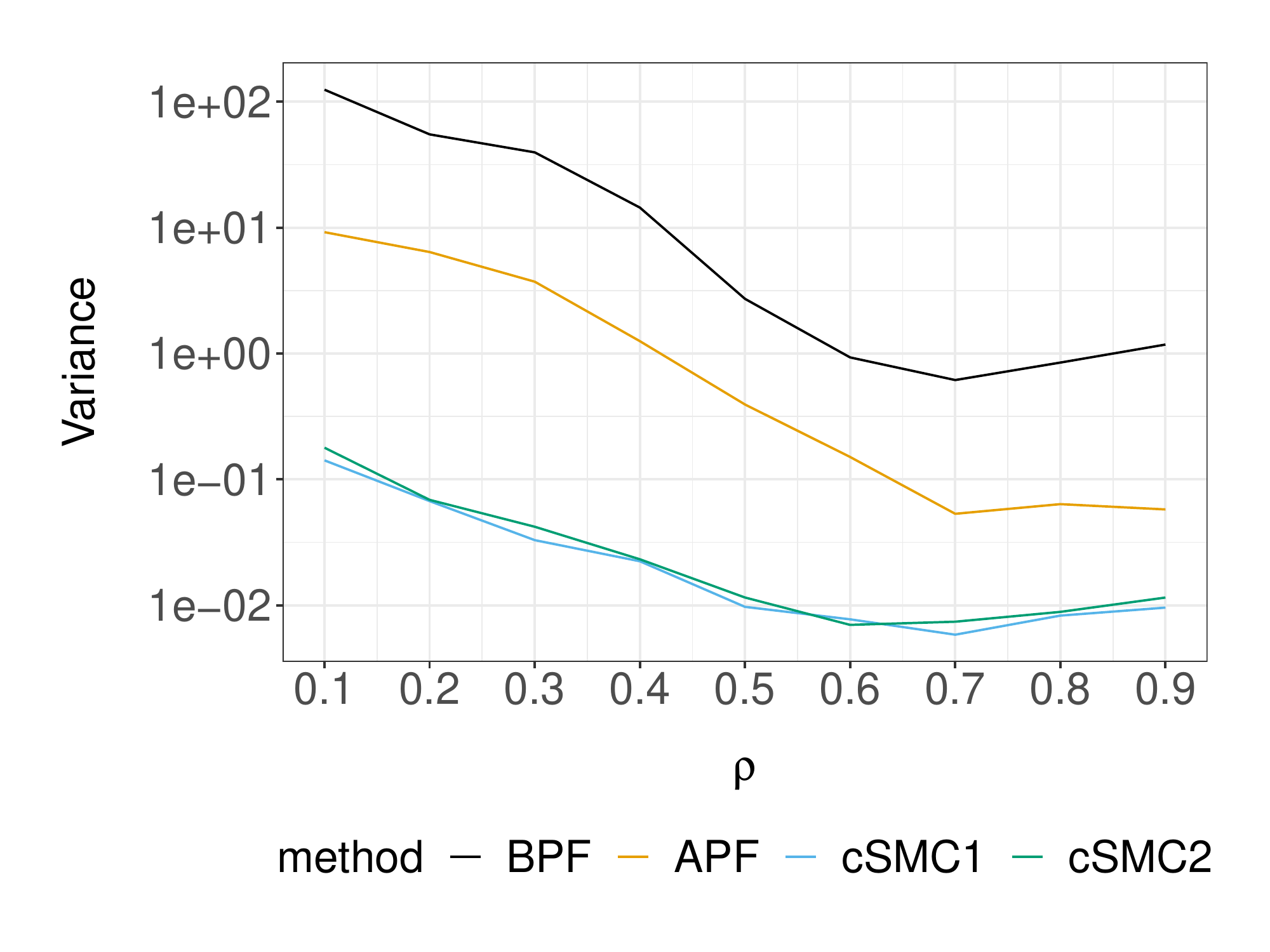}\includegraphics[width=0.5\textwidth]{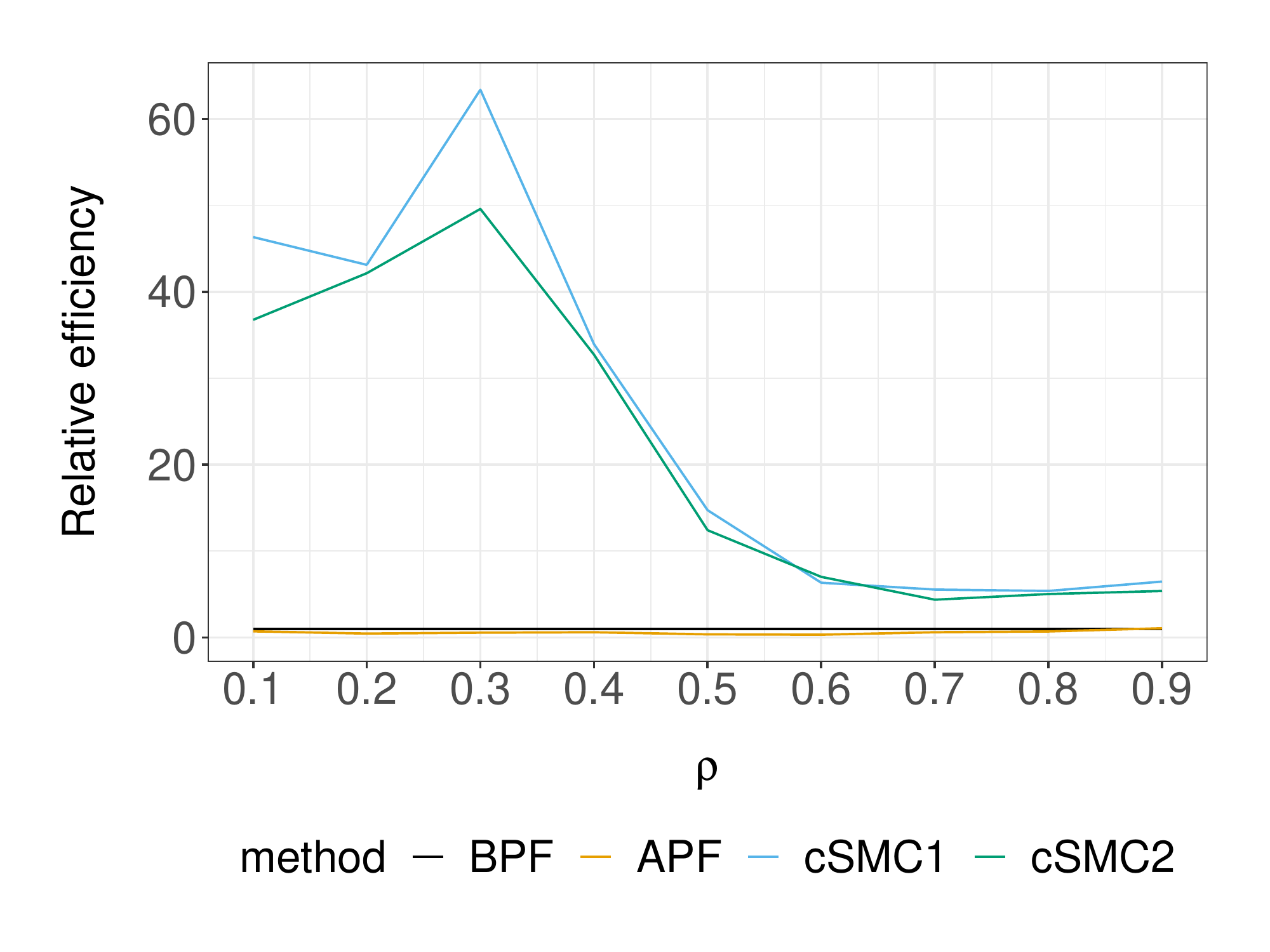}
\caption{Variance of log-marginal likelihood estimator (\emph{left}) and relative efficiency (\emph{right}) of various SMC methods compared 
to the BPF with $P=512$ particles, as $\rho$ varies and the other parameters fixed at the DGP. }
\label{fig:sis_viance_rho}
\end{figure}

\subsection{Parameter and state inference}
\label{sub:sis_parameter_inference}

We first concern ourselves with the behaviour of the marginal likelihood 
and the maximum likelihood estimator (MLE) $\arg\max_{\theta\in\Theta}p_{\theta}(y_{0:T})$ 
as the number of observations $T\rightarrow\infty$. 
This is illustrated with our running simulated dataset from Section~\ref{sub:sis_illustrate_smc}.
Figure~\ref{fig:sis_likelihood_hetero} plots the log-likelihood as a function 
of $\beta_{\lambda}=(\beta_{\lambda}^1,\beta_{\lambda}^2)$ 
or $(\beta_{\lambda}^2,\beta_{\gamma}^2)$ with the other parameters fixed at their data generating values, estimated using cSMC with $P=64$ particles. 
These plots reveal the complex behavior of the likelihood functions induced by agent-based SIS models. 
Furthermore, we see that the likelihood concentrates more around the DGP as $T$ increases, 
and that the MLE can recover the DGP when $T$ is sufficiently large. 

\begin{figure}[htbp]
\centering
        \includegraphics[width=1.0\textwidth]{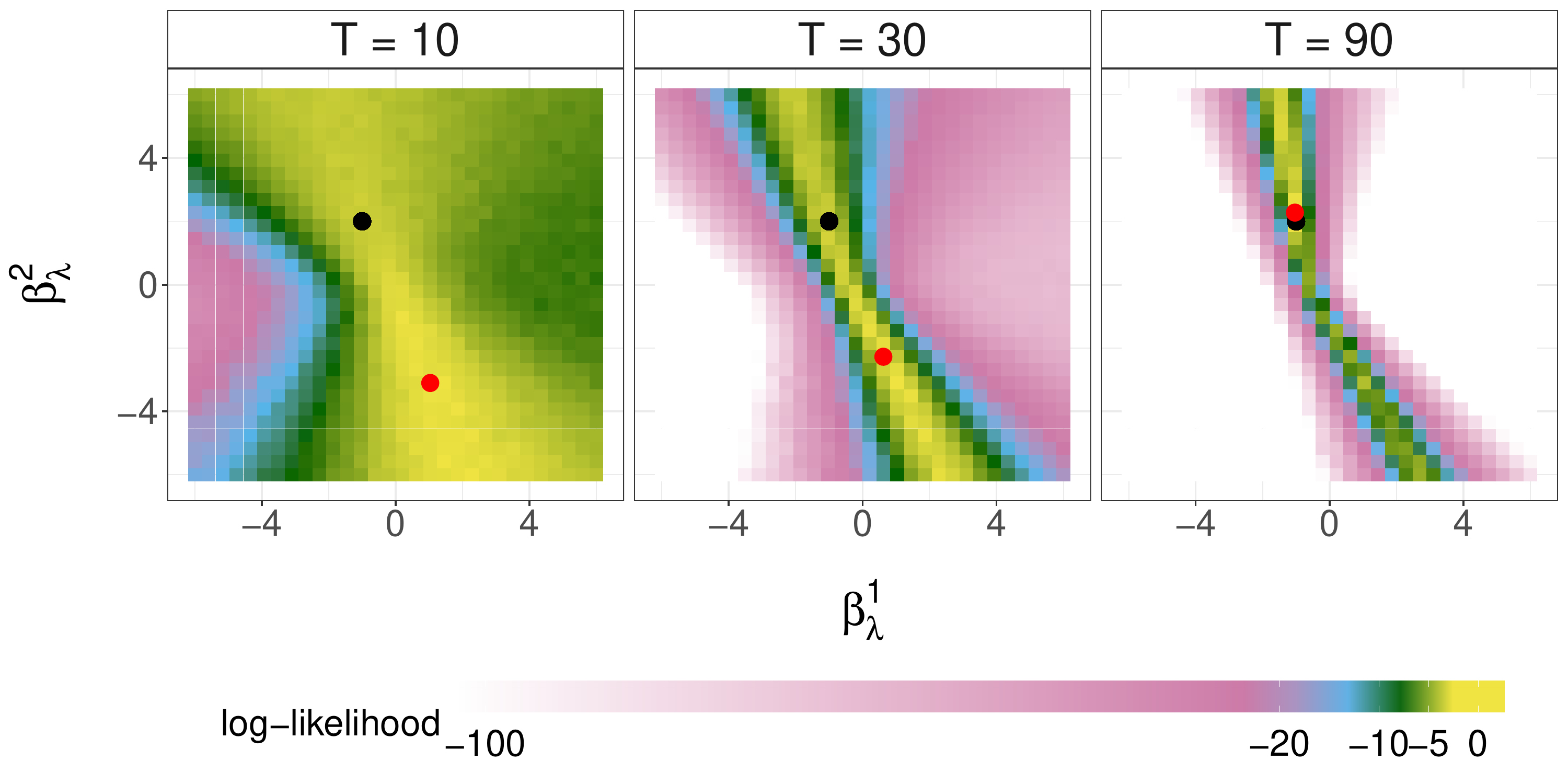}
        \includegraphics[width=1.0\textwidth]{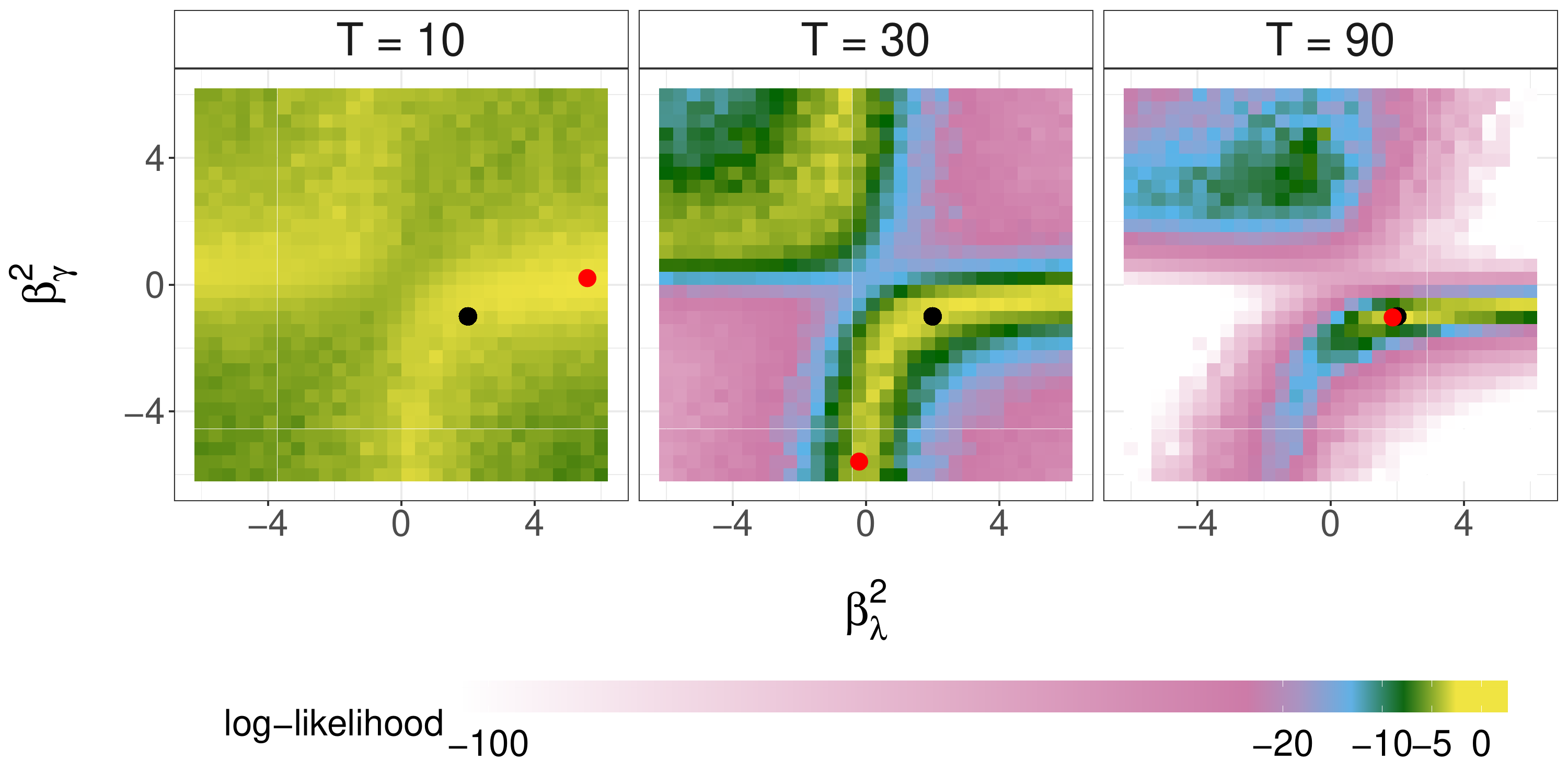}
\caption{Estimated log-likelihood as a function of $\beta_{\lambda}=(\beta_{\lambda}^1,\beta_{\lambda}^2)$ (\emph{first row})
or $(\beta_{\lambda}^2,\beta_{\gamma}^2)$ (\emph{second row}) 
with the other parameters fixed at the DGP given $t\in\{10,30,90\}$ observations. 
For ease of plotting, the log-likelihood values were translated so that the MLE (\emph{red dot}) attains a maximum of zero. 
The data generating parameters of $\beta_{\lambda}=(-1,2)$ or $(\beta_{\lambda}^2,\beta_{\gamma}^2)=(2,-1)$  are shown as a black dot.}
\label{fig:sis_likelihood_hetero}
\end{figure}

By building on the SMC methods in Sections \ref{sub:sis_apf} and \ref{sub:sis_controlled_smc},
one can construct a stochastic gradient ascent scheme 
or an expectation-maximization algorithm to approximate the MLE.
We refer readers to \citet{kantas2015particle} for a comprehensive review of such approaches. 
In the Bayesian framework, our interest is on the posterior distribution 
\begin{align}\label{eqn:sis_joint_posterior}
	p(\theta,x_{0:T}|y_{0:T}) = p(\theta|y_{0:T}) p_{\theta}(x_{0:T}|y_{0:T}) 
	\propto p(\theta)p_{\theta}(x_{0:T},y_{0:T}),
\end{align}
where $p(d\theta)=p(\theta)d\theta$ is a given prior distribution on the parameter space $\Theta$. 
We employ particle marginal Metropolis--Hastings (PMMH) to sample from the posterior distribution. 
Following the discussion in Section \ref{sub:sis_illustrate_smc}, we will choose cSMC to construct a more efficient PMMH chain.

We now illustrate our inference method on the simulated data setup of Section \ref{sub:sis_illustrate_smc}. 
We adopt a prior distribution that assumes the parameters are independent with 
$\beta_0,\beta_{\lambda},\beta_{\gamma}\sim\mathcal{N}(0,9)$ and $\rho\sim\text{Uniform}(0,1)$. 
All MH parameter updates employ a Normal random walk proposal transition on 
the $(\beta_0,\beta_{\lambda},\beta_{\gamma},\log(\rho/(1-\rho)))$-space, 
with a standard deviation of $0.2$ %for PMMH %and $0.08$ for Gibbs samplers 
to achieve suitable acceptance probabilities. 
We use $P=128$ particles in the cSMC algorithm within PMMH and
we run the PMMH algorithm for $100,000$ iterations after a burn in of $5000$ iterations. 
Using these posterior samples, we infer the ratios $R_0^n = \lambda^n / \gamma^n$, 
which can be understood as the reproductive number of agent $n\in[1:N]$. 
In the left panel of Figure~\ref{fig:sis_pmcmc_posterior_envelope}, 
we display the estimated posterior medians and $95\%$ credible sets,  
as well as the data generating values. 
Although the posterior median estimates are similar across agents, 
there is large posterior uncertainty for agents with small or large 
data generated ratios. 
To visualize how $(R_0^n)_{n\in[1:N]}$ is distributed in the population, 
we estimate histograms that take parameter uncertainty into account, 
illustrate the results in the right panel of Figure~\ref{fig:sis_pmcmc_posterior_envelope}. 
The posterior median estimates yields a histogram that is more concentrated 
than its data generating counterpart. 

\begin{figure}[htbp]
	\centering
	\includegraphics[width=.5\textwidth]{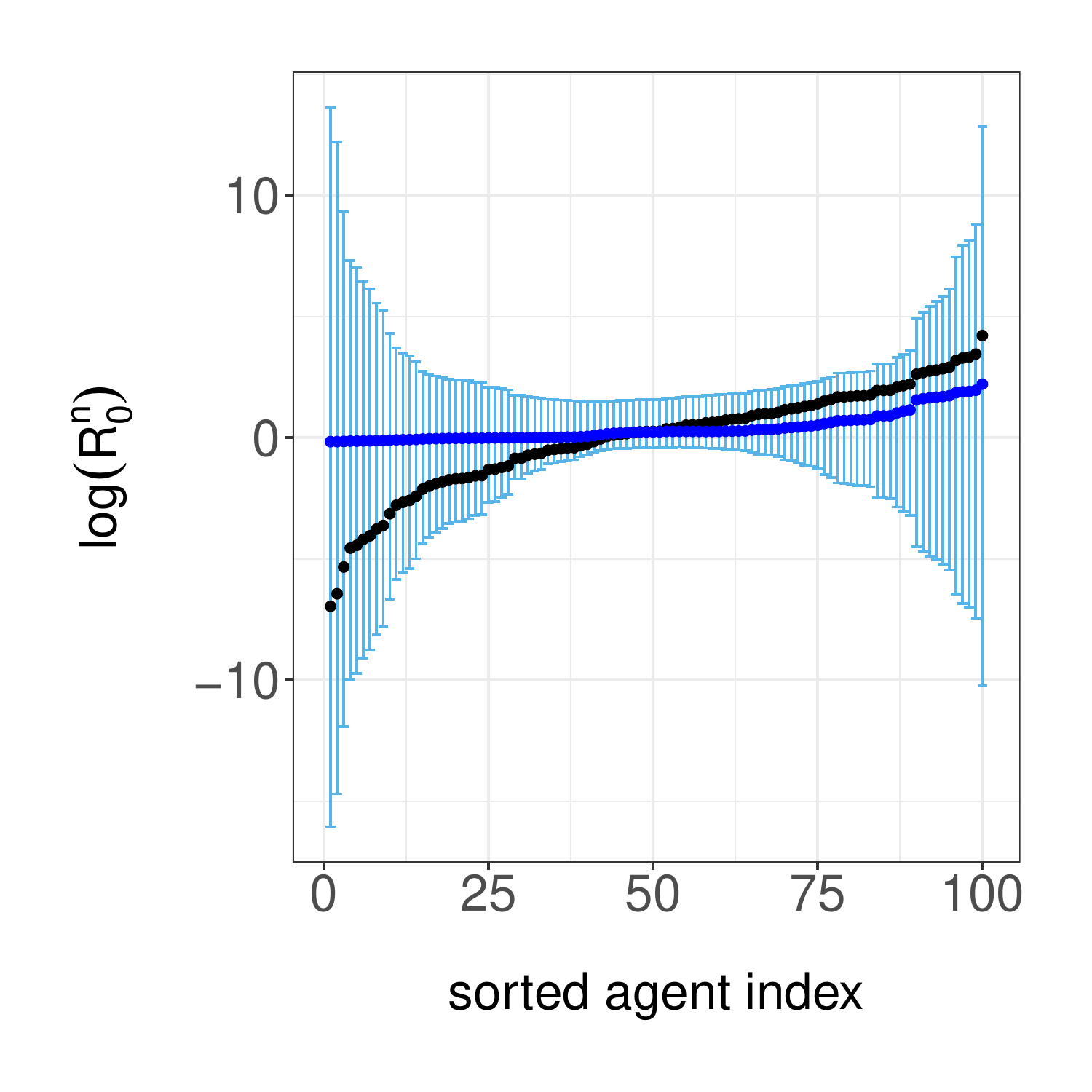}\includegraphics[width=.5\textwidth]{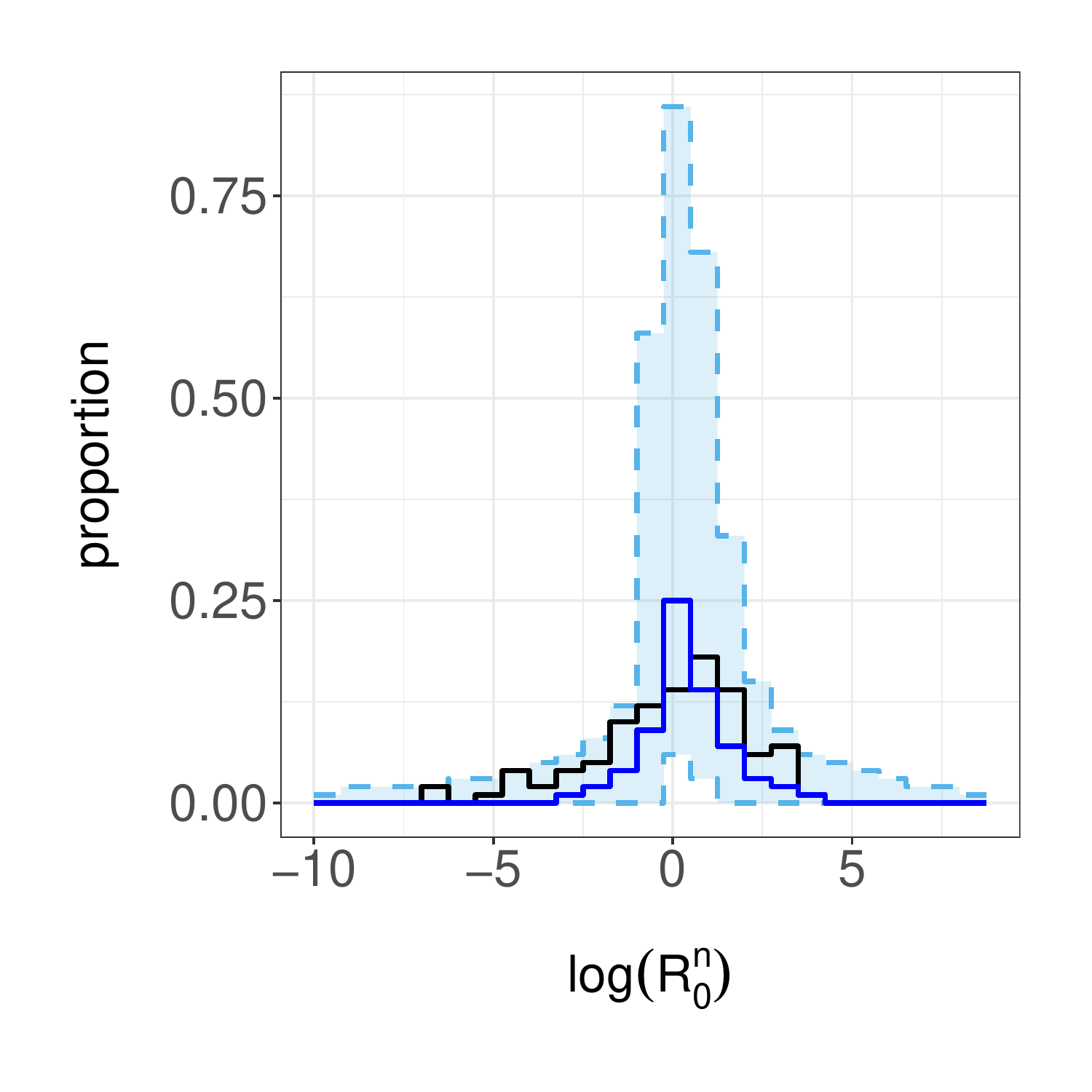}
	\caption{Posterior estimates of the reproductive number $R_0^n = \lambda^n / \gamma^n$ of each agent $n\in[1:N]$ (\emph{left}) 
	and its distribution in the population (\emph{right}). 
	Posterior medians are shown in blue and the $95\%$ posterior credible sets in light blue. 
	Data generating values are illustrated in black.	
	}
	\label{fig:sis_pmcmc_posterior_envelope}
\end{figure}

Lastly, we examine the predictive performance of the model when relatively few observations are available. 
As illustrated in Figure~\ref{fig:sis_prediction}, we assume access to the first $t=30$ observations (black dots) 
and predict the rest of the time series up to time $T=90$ (grey dots) using the posterior predictive distribution 
$p(y_{t+1:T}|y_{0:t}) = \int_{\Theta}\sum_{x_t\in\{0,1\}^N} p_{\theta}(y_{t+1:T}|x_t) p(\theta,x_t|y_{0:t})$. 
By simulating trajectories from the posterior predictive (grey lines), we obtain the model predictions and
uncertainty estimates in Figure~\ref{fig:sis_prediction}. 

\begin{figure}[htbp]
	\centering
	\includegraphics[width = \textwidth]{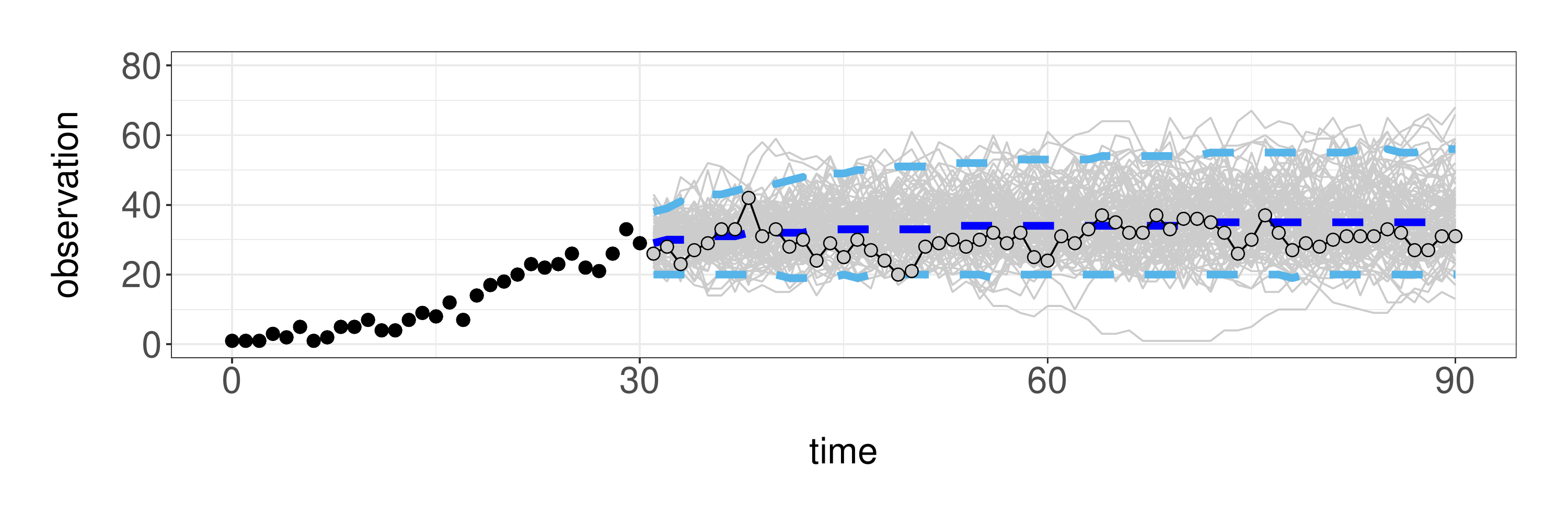}
	\caption{Observation sequence plotted as dots that are coloured initially in black followed by  
	grey from time $t=30$ to $T=90$. 
	Some sampled trajectories from the posterior predictive distribution $p(y_{t+1:T}|y_{0:t})$ 
	are depicted as grey lines. 
	The blue dashed line shows the medians over time; the lower and upper light blue dashed lines 
	correspond to the $2.5\%$ and $97.5\%$ quantiles respectively. }
	\label{fig:sis_prediction}
\end{figure}

\section{Susceptible-Infected-Recovered model}\label{sec:sir}
% \subsection{Model specification}\label{sub:sir_model}

We consider a susceptible-infected-recovered (SIR) model, 
where agents become immune to a pathogen after they recover from an infection. 
The ``recovered'' status of an agent shall be encoded by a state of $2$. 
Given a population configuration $x=(x^n)_{n\in[1:N]}\in\{0,1,2\}^N$, 
we will write 
$S^n(x)=\mathbbm{1}(x^n=0)$, $I^n(x)=\mathbbm{1}(x^n=1)$, $R^n(x)=\mathbbm{1}(x^n=2)$ 
to indicate the status of agent $n\in[1:N]$, and $S(x)=\sum_{n=1}^NS^n(x)$, $I(x)=\sum_{n=1}^NI^n(x)$, 
$R(x)=\sum_{n=1}^NR^n(x)$ to count the number of agents in each state. 
Under the assumption of a closed population, we have $S(x)+I(x)+R(x)=N$. 

The time evolution of the population $(X_t)_{t\in[0:T]}$ is now modelled by a Markov chain on $\{0,1,2\}^N$, 
i.e. the specification in \eqref{eqn:dynamic_X} with 
\begin{align}\label{eqn:sir_dynamic_X}
  \mu_{\theta}(x_0) = \prod_{n=1}^N\Cat(x_0^n;[0:2],\alpha_0^n), \quad  
  f_{\theta}(x_t |x_{t-1})= \prod_{n=1}^N\Cat(x_t^n;[0:2],\alpha^n(x_{t-1})),\quad t\in[1:T]. 
\end{align}
The above probabilities are given by $\alpha_0^n = (\alpha_{0,S}^n, \alpha_{0,I}^n, \alpha_{0,R}^n) = 
(1-\alpha_{0,I}^n, \alpha_{0,I}^n, 0)$ 
and $\alpha^n(x_{t-1}) = (\alpha_S^n(x_{t-1}), \alpha_I^n(x_{t-1}), \alpha_R^n(x_{t-1}))$ where 
\begin{align}\label{eqn:sir_prob}
  &\alpha_S^n(x_{t-1}) = S^n(x_{t-1})\left(1- \lambda^n \mathcal{D}(n)^{-1}\sum_{m\in\mathcal{N}(n)}I^m(x_{t-1})\right), \notag\\
  &\alpha_I^n(x_{t-1}) = S^n(x_{t-1})\lambda^n \mathcal{D}(n)^{-1}\sum_{m\in\mathcal{N}(n)}I^m(x_{t-1}) + I^n(x_{t-1})(1-\gamma^n), \\
  &\alpha_R^n(x_{t-1}) = I^n(x_{t-1})\gamma^n + R^n(x_{t-1}),\notag
\end{align}
which satisfies $\alpha_S^n(x_{t-1}) + \alpha_I^n(x_{t-1}) + \alpha_R^n(x_{t-1}) = 1$. 
Note that $\alpha_R^n(x_{t-1})=1$ if $R^n(x_{t-1}) = 1$, which reflects the above-mentioned immunity. 
%We also notice that at most 2 values in $\alpha^n(x_{t-1})$ is non-zero. 
%We will later use this property in Section \ref{sub:sir_apf} and \ref{sub:sir_state_inference}.
Just like the SIS model, the agents' initial infection probabilities, infection and recovery rates $(\alpha_{0,I}^n,\lambda^n,\gamma^n)_{n\in[1:N]}$ 
are specified with \eqref{eqn:latentprocess_parameters}, and the observation model for the 
number of infections reported over time $(Y_t)_{t\in[0:T]}$ is~\eqref{eqn:static_Y}. 
% Therefore the unknown parameters of the SIR model are also 
% $\theta=(\beta_0,\beta_{\lambda},\beta_{\gamma},\rho)\in\Theta = \mathbb{R}^{3d}\times(0,1)$.

We consider again SMC approximations of the marginal likelihood %\eqref{eqn:marginal_likelihood} 
and smoothing distribution %\eqref{eqn:smoothing_distribution_bayes} 
of the resulting hidden Markov model on $\{0,1,2\}^N$. 
The BPF can be readily implemented in $\bigo(NTP)$ cost, 
but suffers from the difficulties discussed in Section \ref{sub:sis_apf}. 
To obtain better performance, we discuss how to implement the 
fully adapted APF in Section \ref{sub:sir_apf}
and adapt our cSMC construction in Section \ref{sub:sir_csmc}. 
Alternative MCMC approaches such as those detailed in Appendix~\ref{sub:mcmc} 
could also be considered. 

\subsection{Auxiliary particle filter}
\label{sub:sir_apf}
We recall that implementing the fully adapted APF requires one to sample from proposals 
$q_0(x_0|\theta)=p_{\theta}(x_0|y_0), q_t(x_t|x_{t-1},\theta)=p_{\theta}(x_t|x_{t-1},y_t)$ for $t\in[1:T]$, and 
evaluate the weight functions $w_0(x_0)=p_{\theta}(y_0)$ and $w_t(x_{t-1})=p_{\theta}(y_t|x_{t-1})$ for $t\in[1:T]$. 
 At the initial time, we compute the marginal likelihood $p_{\theta}(y_0)$ as 
 \begin{align}
  p_{\theta}(y_0) = \sum_{i_0=0}^N \PB(i_0;\alpha_{0,I}) \Binomial(y_0; i_0, \rho) \mathbbm{1}(i_0\geq y_0)
 \end{align}
 where $\alpha_{0,I}=(\alpha_{0,I}^n)_{n\in[1:N]}$, and sample from 
 \begin{align}\label{eqn:sir_apf_initial}
  &p_{\theta}(x_0,i_0,i_0^{1:N}|y_0) = p_{\theta}(i_0|y_0)p_{\theta}(i_0^{1:N}|i_0)p_{\theta}(x_0|i_0^{1:N}) \\
  &= \frac{\PB(i_0;\alpha_{0,I}) \Binomial(y_0; i_0, \rho) 
  \mathbbm{1}(i_0 \geq y_0)}{p_{\theta}(y_0)}\CB(i_0^{1:N};\alpha_{0,I}, i_0)\prod_{n=1}^N\Cat(x_0^n;[0:2],\nu_0(i_0^n))\notag 
 \end{align}
 with $\nu_0(i_0^n)=(1-i_0^n,i_0^n,0)$, which admits the posterior of agent states $p_{\theta}(x_0|y_0)$ as its marginal distribution. 
 Similarly, for time step $t\in[1:T]$, we compute the predictive likelihood $p_{\theta}(y_t|x_{t-1})$ as 
 \begin{align}
  p_{\theta}(y_t|x_{t-1}) = \sum_{i_{t}=0}^N \PB(i_{t};\alpha_I(x_{t-1})) \Binomial(y_t; i_{t}, \rho) \mathbbm{1}(i_{t}\geq y_t)
 \end{align}
 where $\alpha_I(x_{t-1})=(\alpha_I^n(x_{t-1}))_{n\in[1:N]}$, and sample from 
 \begin{align}\label{eqn:sir_apf_transition}
  &p_{\theta}(x_t,i_t,i_t^{1:N}|x_{t-1},y_t) = p_{\theta}(i_t|x_{t-1},y_t)p_{\theta}(i_t^{1:N}|x_{t-1},i_t)p_{\theta}(x_t|x_{t-1},i_t^{1:N}) \\
  & = \frac{\PB(i_t;\alpha_I(x_{t-1})) \Binomial(y_t; i_t, \rho) 
  \mathbbm{1}(i_t \geq y_t)}{p_{\theta}(y_t|x_{t-1})}\CB(i_t^{1:N};\alpha_I(x_{t-1}), i_t)
  \prod_{n=1}^N\Cat(x_t^n;[0:2],\nu(x_{t-1}, i_t^n))\notag 
 \end{align}
 with  
 $\nu(x_{t-1}, i_t^n) = ((1-i_t^n)S^n(x_{t-1}),i_t^n,(1-i_t^n)\{I^n(x_{t-1}) + R^n(x_{t-1})\})$,
 which admits $p_{\theta}(x_t|x_{t-1},y_t)$ as its marginal transition. 

By augmenting the infection status of the agents $I^{1:N}(x_t)=(I^n(x_t))_{n\in[1:N]}$, 
the first two steps in \eqref{eqn:sir_apf_initial} and \eqref{eqn:sir_apf_transition} are analogous to 
\eqref{eqn:sis_apf_initial} and \eqref{eqn:sis_apf_transition} for the SIS model, 
and the last step is to identify an agent's state given her augmented current infection status and previous state.
We point out that $\nu_0(i_0^n), \nu(x_{t-1},i_t^n) \in \{(1,0,0),(0,1,0),(0,0,1)\}.$
These expressions for $\nu_0(i_0^n)$ and $\nu(x_{t-1},i_t^n)$ exploit the following facts: 
($i$) susceptible agents either remain susceptible or become infected; 
($ii$) infected agents either remain infected or recover; 
($iii$) agents who have recovered enjoy immunity.  
 % Under these properties of the SIR model, 
 % introducing either the susceptibility or recovery status of the agents as auxiliary variables 
 % would not allow us to uniquely identify the next states of the agents. 
Algorithm \ref{algo:apf_sir} provides an algorithmic description of the resulting APF, 
which has the same cost as APF for the SIS model. 

\begin{algorithm}
  \SetAlgoLined
  \KwIn{Parameters $\theta\in\Theta$ and number of particles $P\in\mathbb{N}$} 
  compute $v_0^{(i)} = \PB(i;\alpha_{0,I}) \Binomial(y_0; i, \rho) \mathbbm{1}(i\geq y_0)$ for $i\in[0:N]$\\ 
  set $w_0 = \sum_{i=0}^Nv_0^{(i)}$\\ 
  normalize $V_0^{(i)} = v_0^{(i)} / w_0$ for $i\in[0:N]$\\   
  sample $I_0^{(p)}\sim\Cat([0:N],V_0^{(0:N)})$, $I_0^{1:N,(p)} | I_0^{(p)} \sim\CB(\alpha_{0,I}, I_0^{(p)})$ for $p\in[1:P]$\\
  sample $X_0^{n,(p)} | I_0^{n,(p)}\sim\Cat([0:2],\nu_0(I_0^{n,(p)}))$ for $p\in[1:P]$ and $n\in[1:N]$\\
  \For {$t = 1,\cdots, T$ and $p = 1,\cdots, P$}{
    compute $v_t^{(i,p)} = \PB(i;\alpha_I(X_{t-1}^{(p)})) \Binomial(y_t; i, \rho) \mathbbm{1}(i\geq y_t)$ for $i\in[0:N]$\\     
    set $w_t^{(p)} = \sum_{i=0}^Nv_t^{(i,p)}$\\       
    normalize $V_t^{(i,p)} = v_t^{(i,p)} / w_t^{(p)}$ for $i\in[0:N]$\\   
    normalize $W_{t}^{(p)} = w_t^{(p)} / \sum_{k=1}^Pw_t^{(k)}$\\
    sample $A_{t-1}^{(p)}\sim r(\cdot|W_{t}^{(1:P)})$\\
    sample $I_t^{(p)}\sim\Cat([0:N],V_t^{(0:N,A_{t-1}^{(p)})})$ and $I_t^{1:N,(p)} | I_t^{(p)} \sim\CB(\alpha_I(X_{t-1}^{(A_{t-1}^{(p)})}), I_t^{(p)})$\\
    sample $X_t^{n,(p)} \mid X_{t-1}^{(A_{t-1}^{(p)})}, I_t^{n,(p)}\sim \Cat([0:2],\nu(X_{t-1}^{(A_{t-1}^{(p)})}, I_t^{n,(p)}))$ for $n\in[1:N]$
  }
  
  \KwOut{Marginal likelihood estimator $\hat{p}_{\theta}(y_{0:T})=w_0\prod_{t=1}^TP^{-1}\sum_{p=1}^Pw_t^{(p)}$, states 
  $(X_t^{(p)})_{(t,p)\in[0:T]\times[1:P]}$ and ancestors $(A_t^{(p)})_{(t,p)\in[0:T-1]\times[1:P]}$}
  \caption{Auxiliary particle filter for SIR model}\label{algo:apf_sir}
\end{algorithm}

\subsection{Controlled sequential Monte Carlo}
\label{sub:sir_csmc}
We now consider approximation of the BIF~\eqref{eqn:backward_info_filter} to construct a proposal 
distribution approximating~\eqref{eqn:decompose_smoothing}. 
At the terminal time $T$, it suffices to compute
$\psi_T(i_T)=\Binomial(y_T;i_T, \rho)\mathbbm{1}(i_T \geq y_T)$ for all $i_T\in[0:N]$ to represent 
the function $\psi_T^\star(x_T)$. 
As before, the next iterate $\psi_{T-1}^\star(x_{T-1})$ requires an approximation of
the conditional expectation $f_{\theta}(\psi_{T}|x_{T-1})= \sum_{i_T=0}^N \PB(i_T;\alpha_I(x_{T-1})) \psi_T(i_T)$. 
Following the arguments in \eqref{eqn:sis_alpha_approx}, we approximate $\alpha(x_{T-1})$ by 
\begin{align*}
  \bar{\alpha}^n(x_{T-1}) = (\bar{\alpha}_S^n(x_{T-1}), \bar{\alpha}_I^n(x_{T-1}), \bar{\alpha}_R^n(x_{T-1})),
\end{align*} defined as  
\begin{align}\label{eqn:sir_alpha_approx}
  &\bar{\alpha}_S^n(x_{T-1}) = S^n(x_{T-1})\left(1- \bar{\lambda} N^{-1}I(x_{T-1})\right), \notag\\
  &\bar{\alpha}_I^n(x_{T-1}) = S^n(x_{T-1})\bar{\lambda} N^{-1}I(x_{T-1}) + I^n(x_{T-1})(1-\bar{\gamma}), \\
  &\bar{\alpha}_R^n(x_{T-1}) = I^n(x_{T-1})\bar{\gamma} + R^n(x_{T-1}),\notag
\end{align}
which satisfies $\bar{\alpha}_S^n(x_{T-1}) + \bar{\alpha}_I^n(x_{T-1}) + \bar{\alpha}_R^n(x_{T-1}) = 1$. 
Writing the corresponding Markov transition as 
$\bar{f}_{\theta}(x_T|x_{T-1})=\prod_{n=1}^N\Cat(x_T^n;[0:2],\bar{\alpha}^n(x_{T-1}))$, 
we approximate the conditional expectation $f_{\theta}(\psi_{T}|x_{T-1})$ by 
\begin{align}\label{eqn:sir_condexp_approx}
  \bar{f}_{\theta}(\psi_{T}|S(x_{T-1}),I(x_{T-1})) = \sum_{i_T=0}^N 
  \SB(i_T; S(x_{T-1}), \bar{\lambda} N^{-1}I(x_{T-1}), I(x_{T-1}), 1-\bar{\gamma})\psi_T(i_T).
\end{align}
Although \eqref{eqn:sir_condexp_approx} is analogous to \eqref{eqn:sis_condexp_approx} for the SIS model, 
the number of susceptible agents $S(x_{T-1})$ cannot be determined by just knowing the number of infections 
$I(x_{T-1})$ in the SIR model. 
Therefore it is necessary to account for both $S(x_{T-1})$ and $I(x_{T-1})$ in our approximation, i.e. we compute 
\begin{align}
  \psi_{T-1}(s_{T-1},i_{T-1})=\Binomial(y_{T-1};i_{T-1}, \rho)\mathbbm{1}(i_{T-1} \geq y_{T-1})\bar{f}_{\theta}(\psi_{T}|s_{T-1},i_{T-1})
\end{align}
for all $(s_{T-1},i_{T-1})\in[0:N]\times[0:(N-s_{T-1})]$. 
Since the number of agents in the population is fixed, it is also possible to 
work with the variables $I(x_{T-1})$ and $R(x_{T-1})$ instead. 
Subsequently for $t\in[0:T-2]$, we approximate \eqref{eqn:backward_info_filter} by 
\begin{align}
  \psi_{t}(s_{t},i_{t})=\Binomial(y_{t};i_{t}, \rho)\mathbbm{1}(i_{t} \geq y_{t})\bar{f}_{\theta}(\psi_{t+1}|s_{t},i_{t}),
  \quad (s_{t},i_{t})\in[0:N]\times[0:(N-s_{t})],
\end{align}
where $\bar{f}_{\theta}(\psi_{t+1}|s_t,i_t) = \sum_{s_{t+1}=0}^N\sum_{i_{t+1}=0}^{N-s_{t+1}} 
\bar{f}_{\theta}(s_{t+1},i_{t+1}|s_t,i_t)\psi_{t+1}(s_{t+1},i_{t+1})$, 
\begin{align}\label{eqn:sir_summary_transition}
  \bar{f}_{\theta}(s_{t+1},i_{t+1}|s_t,i_t) = 
  \Binomial(s_{t+1};s_t,1-\bar{\lambda}N^{-1}i_t)\Binomial(i_t-i_{t+1}+s_t-s_{t+1};i_t,\bar{\gamma}),
\end{align}
for $(s_{t+1},i_{t+1})\in[0:s_t]\times[(s_t-s_{t+1}):(i_t+s_t-s_{t+1})]$, and zero otherwise. 
The above expression follows from the SIR model structure under homogeneous probabilities \eqref{eqn:sir_alpha_approx}. 
Algorithm \ref{algo:bif_sir} summarizes our approximation of the BIF $(\psi_t)_{t\in[0:T]}$, 
which costs $\bigo(N^4T)$ to compute and $\bigo(N^2T)$ in storage. 

\begin{algorithm}
	\SetAlgoLined
	\KwIn{Parameters $\theta\in\Theta$} 
	  compute $\psi_T(i_T)=\Binomial(y_T;i_T, \rho)\mathbbm{1}(i_T \geq y_T)$ for $i_T\in[0:N]$\\
  \For {$t = T-1,\cdots, 0$, $s_t = 0, \ldots, N$ and $i_t = 0,\cdots, N-s_t$}{
        \eIf {$t=T-1$}{
      compute $ \bar{f}_{\theta}(\psi_{t+1}|s_{t},i_{t}) = \sum_{i_{t+1}=0}^N
      \SB(i_{t+1}; s_{t}, \bar{\lambda} N^{-1}i_{t}, i_{t}, 1-\bar{\gamma})\psi_{t+1}(i_{t+1})$\\
    }{
      compute $ \bar{f}_{\theta}(\psi_{t+1}|s_{t},i_{t}) = \sum_{s_{t+1}=0}^N\sum_{i_{t+1}=0}^{N-s_{t+1}} 
      \bar{f}_{\theta}(s_{t+1},i_{t+1}|i_t,s_t)\psi_{t+1}(s_{t+1},i_{t+1})$\\
    }
  compute $ \psi_{t}(s_{t},i_{t})=\Binomial(y_{t};i_{t}, \rho)\mathbbm{1}(i_{t} \geq y_{t})\bar{f}_{\theta}(\psi_{t+1}|s_{t},i_{t})$
  } 
	\KwOut{Approximate BIF $(\psi_{t})_{t\in[0:T]}$}
	\caption{Backward information filter approximation for SIR model}\label{algo:bif_sir}
\end{algorithm}

We can define our proposal distribution and SMC weight functions 
in the same manner as \eqref{eqn:csmc_proposals} 
and \eqref{eqn:csmc_weights}, respectively. 
Expectations appearing in these SMC weights can be computed as
\begin{align}
  \mu_{\theta}(\psi_0) = \sum_{i_0=0}^N\PB(i_0;\alpha_{0,I})\psi_0(N-i_0,i_0),\quad
  f_{\theta}(\psi_{t}|x_{t-1}) = \sum_{s_t=0}^N\sum_{i_t=0}^{N-s_t}f_{\theta}(s_t,i_t|x_{t-1})\psi_t(s_t,i_t), 
\end{align}
for $t\in[1:T-1]$ and $f_{\theta}(\psi_T|x_{T-1})=\sum_{i_T=0}^N\PB(i_T;\alpha_I(x_{T-1}))\psi_T(i_T)$, where
\begin{align}
  f_{\theta}(s_t,i_t|x_{t-1}) = \PB(s_t;\alpha_S(x_{t-1}))\PB(I(x_{t-1})-i_{t}+S(x_{t-1})-s_t;(I^n(x_{t-1})\gamma^n)_{n\in[1:N]}),
\end{align}
for $(s_{t},i_{t})\in[0:S(x_{t-1})]\times[(S(x_{t-1})-s_{t}):(I(x_{t-1})+S(x_{t-1})-s_{t})]$, and zero otherwise. 
The above expression should be understood as the analogue of \eqref{eqn:sir_summary_transition} 
under the heterogeneous probabilities \eqref{eqn:sir_prob}. 
Sampling from the proposals can be done in a similar manner as the APF in Section \ref{sub:sir_apf}.
At the initial time, we sample from 
\begin{align}
  &q_0(x_0,i_0,i_0^{1:N}|\theta) = q_0(i_0|\theta)q_0(i_0^{1:N}|i_0,\theta)q_0(x_0|i_0^{1:N},\theta) \\
  &= \frac{\PB(i_0;\alpha_{0,I})\psi_0(N-i_0,i_0)}{\mu_{\theta}(\psi_0)}\CB(i_0^{1:N};\alpha_{0,I}, i_0)\prod_{n=1}^N\Cat(x_0^n;[0:2],\nu_0(i_0^n)),\notag 
\end{align}
which admits $q_0(x_0|\theta)$ as its marginal distribution. 
For time $t\in[1:T-1]$, we sample from 
\begin{align}
  &q_t(x_t,i_t,i_t^{1:N}|x_{t-1},\theta) = q_t(i_t|x_{t-1},\theta)q_t(i_t^{1:N}|x_{t-1},i_t,\theta)
  q_t(x_t|x_{t-1},i_t^{1:N},\theta) \\
  & = \frac{\sum_{s_t=0}^{N-i_t}f_{\theta}(s_t,i_t|x_{t-1})\psi_t(s_t,i_t)}{f_{\theta}(\psi_t|x_{t-1})}\CB(i_t^{1:N};\alpha_I(x_{t-1}), i_t)
  \prod_{n=1}^N\Cat(x_t^n;[0:2],\nu(x_{t-1}, i_t^n)),\notag 
\end{align}
and 
\begin{align}
  &q_T(x_T,i_T,i_T^{1:N}|x_{T-1},\theta) = q_T(i_T|x_{T-1},\theta)q_T(i_T^{1:N}|x_{T-1},i_T,\theta)
  q_T(x_T|x_{T-1},i_T^{1:N},\theta) \\
  & = \frac{\PB(i_T;\alpha_I(x_{T-1}))\psi_T(i_T)}{f_{\theta}(\psi_T|x_{T-1})}\CB(i_T^{1:N};\alpha_I(x_{T-1}), i_T)
  \prod_{n=1}^N\Cat(x_T^n;[0:2],\nu(x_{T-1}, i_T^n)),\notag 
\end{align}
which admits $q_t(x_t|x_{t-1},\theta)$ as its marginal transition for all $t\in[1:T]$. 
Algorithm \ref{algo:csmc_sir} details the resulting cSMC method, which costs 
the same as cSMC for the SIS model. 

\begin{algorithm}
  \SetAlgoLined
	\KwIn{Parameters $\theta\in\Theta$, approximate BIF $(\psi_{t})_{t\in[0:T]}$ 
	and number of particles $P\in\mathbb{N}$} 	
	compute probabilities $v_0^{(i)} = \PB(i;\alpha_0) \psi_0(i)$ for $i\in[0:N]$\\ 
  compute $v_0^{(i)} = \PB(i;\alpha_{0,I})\psi_0(N-i,i)$ for $i\in[0:N]$\\
  set $E_0 = \sum_{i=0}^Nv_0^{(i)}$\\ 
  normalize $V_0^{(i)} = v_0^{(i)} / E_0$ for $i\in[0:N]$\\   
  sample $I_0^{(p)}\sim\Cat([0:N],V_0^{(0:N)})$ and $I_0^{1:N,(p)} | I_0^{(p)} \sim\CB(\alpha_{0,I}, I_0^{(p)})$ for $p\in[1:P]$\\
  sample $X_0^{n,(p)} | I_0^{n,(p)}\sim\Cat([0:2],\nu_0(I_0^{n,(p)}))$ for $p\in[1:P]$ and $n\in[1:N]$\\
  compute $w_0^{(p)} = w_0(X_0^{(p)})$ using \eqref{eqn:csmc_weights} for $p\in[1:P]$\\
  \For {$t = 1,\cdots, T$ and $p = 1,\cdots, P$}{
    normalize $W_{t-1}^{(p)} = w_{t-1}^{(p)} / \sum_{k=1}^Pw_{t-1}^{(k)}$\\
    sample $A_{t-1}^{(p)}\sim r(\cdot|W_{t-1}^{(1:P)})$\\
    \eIf {$t<T$}{
      compute $v_t^{(i,p)} = \sum_{s=0}^{N-i}f_{\theta}(s,i|X_{t-1}^{(A_{t-1}^{(p)})})\psi_t(s,i)$ for $i\in[0:N]$\\
      }{
      compute $v_t^{(i,p)} = \PB(i;\alpha_I(X_{t-1}^{(A_{t-1}^{(p)})}))\psi_t(i)$ for $i\in[0:N]$\\
      }     
    set $E_t^{(p)}= \sum_{i=0}^Nv_t^{(i,p)}$\\      
    normalize $V_t^{(i,p)} = v_t^{(i,p)} / E_t^{(p)}$ for $i\in[0:N]$\\   
    sample $I_t^{(p)}\sim\Cat([0:N],V_t^{(0:N,p)})$ and $I_t^{1:N,(p)} | I_t^{(p)} \sim\CB(\alpha_I(X_{t-1}^{(A_{t-1}^{(p)})}), I_t^{(p)})$\\
    sample $X_t^{n,(p)} | X_{t-1}^{(A_{t-1}^{(p)})}, I_t^{n,(p)}\sim \Cat([0:2],\nu(X_{t-1}^{(A_{t-1}^{(p)})}, I_t^{n,(p)}))$ for $n\in[1:N]$\\
    compute $w_t^{(p)} = w_t(X_t^{(p)})$ using \eqref{eqn:csmc_weights} \\
  }
  
  \KwOut{Marginal likelihood estimator $\hat{p}_{\theta}(y_{0:T})=\prod_{t=0}^TP^{-1}\sum_{p=1}^Pw_t^{(p)}$, states 
  $(X_t^{(p)})_{(t,p)\in[0:T]\times[1:P]}$ and ancestors $(A_t^{(p)})_{(t,p)\in[0:T-1]\times[1:P]}$}
  \caption{Controlled sequential Monte Carlo for SIR model}\label{algo:csmc_sir}
\end{algorithm}

Our proposal distribution $q_{\theta}(x_{0:T})$ also satisfies 
the Kullback--Leibler upper bound in Proposition \ref{prop:KL_BIF_bound}, 
with appropriate notational extensions to the state space $\{0,1,2\}^N$. 
The following result is analogous to Proposition \ref{prop:BIF_error} for the SIS model.

\begin{proposition}\label{prop:sir_BIF_error}
For each time $t\in[0:T]$, the BIF approximation in Algorithm \ref{algo:bif_sir} satisfies:
\begin{align}
\eta_t^\star(\log(\psi_t^\star/\psi_t)|\theta)\leq \sum_{k=t}^{T-1}c^{\star}_{\theta}(\psi_{k+1})
\{\|\phi^{\star}_{\theta} \|_{L^2(\eta_k^\star)}\|\Delta_{\theta,S}\|_{L^2(\eta_k^\star)} 
+ \|\xi^{\star}_{\theta} \|_{L^2(\eta_k^\star)}\|\Delta_{\theta,R}\|_{L^2(\eta_k^\star)}\},\label{eqn:sir_BIF_logerror1}\\
\eta_t(\log(\psi_t/\psi_t^\star)|\theta)\leq \sum_{k=t}^{T-1}c_{\theta}(\psi_{k+1})
\{\|\phi_{\theta} \|_{L^2(\eta_k)}\|\Delta_{\theta,S}\|_{L^2(\eta_k)} 
+ \|\xi_{\theta} \|_{L^2(\eta_k)}\|\Delta_{\theta,R}\|_{L^2(\eta_k)}\},\label{eqn:sir_BIF_logerror2}
\end{align}
for $t\in[0:T]$. 
The constants are $c^{\star}_{\theta}(\psi_{k})=\{\min_{x_{k-1}\in\{0,1,2\}^N\setminus(0,\ldots,0)}
{f}_{\theta}(\psi_{k}|x_{k-1})\}^{-1}<\infty$, 
$c_{\theta}(\psi_{k})=\{\min_{(s_{k-1},i_{k-1}) \in [0:N-i_{k-1}]\times[1:N]}
\bar{f}_{\theta}(\psi_{k}|s_{k-1},i_{k-1})\}^{-1}<\infty$ for $k\in[1:T]$, and the functions are 
$\Delta_{\theta,S}(x)=\sum_{n=1}^N|\bar{\alpha}_S^n(x)-\alpha_S^n(x)|$,  
$\Delta_{\theta,R}(x)=\sum_{n=1}^N|\bar{\alpha}_R^n(x)-\alpha_R^n(x)|$,  
\begin{align}
  \phi^{\star}_{\theta}(x) = \left\lbrace\sum_{s=0}^{S(x)}\frac{\PB(s;{\alpha}_S(x))^2}{\PB(s;\bar{\alpha}_S(x))^2}\right\rbrace^{1/2}, 
  \quad \xi^{\star}_{\theta}(x) = \left\lbrace \sum_{r=0}^{I(x)}\frac{\PB(r;(I^n(x){\gamma^n})_{n\in[1:N]})^2}
  {\PB(r;(I^n(x)\bar{\gamma})_{n\in[1:N]})^2}\right\rbrace^{1/2}, \\
  \phi_{\theta}(x) = \left\lbrace\sum_{s=0}^{S(x)}\frac{\PB(s;\bar{\alpha}_S(x))^2}{\PB(s;\alpha_S(x))^2}\right\rbrace^{1/2}, 
  \quad \xi_{\theta}(x) = \left\lbrace \sum_{r=0}^{I(x)}\frac{\PB(r;(I^n(x)\bar{\gamma})_{n\in[1:N]})^2}
  {\PB(r;(I^n(x)\gamma^n)_{n\in[1:N]})^2}\right\rbrace^{1/2}, 
\end{align}
for $x\in\{0,1,2\}^N$.
\end{proposition}
The arguments in the proof of Proposition \ref{prop:sir_BIF_error} in Appendix \ref{sec:bif_approx} are similar to 
Proposition \ref{prop:BIF_error}, with some modifications tailored to the SIR model. 
This result shows the dependence of the BIF approximation error on our approximation of 
the transition probability $\alpha(x)$ via the terms $\Delta_{\theta,S}(x)$ and $\Delta_{\theta,R}(x)$. 
As before, we can decompose 
$\Delta_{\theta,S}(x)\leq \Delta_{\theta}^{\mathcal{G}}(x) + \Delta_{\theta}^{\lambda}$ and 
$\Delta_{\theta,R}(x)\leq \Delta_{\theta}^{\gamma}$ into the elements defined 
in \eqref{eqn:Delta_decomposition}. 
One can also obtain more fine-grained approximations using clustering in the spirit of Appendix \ref{sec:finer_BIF_approx}. 

\section{Discussion}
\label{sec:discussion}

Although agent-based models have been widely used as a simulation paradigm, 
statistical inference for these models has not received as much attention, due 
in part to the computational challenges involved. 
We have focused on agent-based models of disease transmission, 
and presented new SMC methods that can estimate their likelihood more efficiently 
than the standard BPF. 

Our proposed methodology can be extended in various directions. 
Instead of the binomial model in \eqref{eqn:static_Y}, other observation models 
such as a negative binomial distribution could be 
considered. We could also adapt the controlled SMC methodology to 
handle the case where observations are only available at a collection of time instances, 
and we can expect further relative gains in such settings. 
Future work could consider settings where the difference in infection counts between 
successive time steps is observed. 

We view our contribution as a step towards inference for larger classes of
agent-based models, and that some of our contributions might be useful beyond
the models considered here.  We hope to motivate further work on alleviating
the computational burden of inference in agent-based models, and that the
removal of some of the computational bottlenecks might encourage further
investigation on the statistical properties of these models. 
%Open questions include identifiability issues and consistency properties of estimators 
%in the limit of the number of observations. 
%Empirical considerations such as the choice of prior distributions is also worth exploring.  

%The agent-based models we have considered achieves the highest level of model complexity, 
%among compartmental models for disease transmission, 
%as their constituent units are comprised of individuals. 
%One could also construct agent-based models with varying levels of aggregation, 
%and explore the resulting bias-variance trade-off. 

%Finally, the agent-based models we present achieves the highest level of model complexity among compartmental models for disease dynamics, since their constituent units are individuals, the smallest unit possible. One can have the flexibility of designing agent-based models with different levels of aggregation by playing with subgroups of the agents. The bias-variance trade-off arising from model complexity remains to be understood. 

\subsubsection*{Acknowledgments}
This work was funded by CY Initiative of Excellence (grant ``Investissements
d'Avenir'' ANR-16-IDEX-0008). 
Pierre E. Jacob gratefully acknowledges support
by the National Science Foundation through grants DMS-1712872 and DMS-1844695.

\small 
\bibliographystyle{plainnat}
\bibliography{ref}

\appendix
\section{Poisson binomial distributions\label{sec:poibin}}

\subsection{Recursive definition of Poisson binomial probabilities\label{sec:recursive-poibin}}
The following provides a derivation of the recursion \eqref{eqn:q-recursion} for the function 
$q(i,n) = \PB(i;\alpha^{n:N}) = \mathbb{P}(\sum_{m=n}^N X^m = i)$ with $X^m\sim\Bernoulli(\alpha^m)$ independently. 
The recursion 
allows the computation of Poisson binomial probabilities in $\mathcal{O}(N^2)$ operations.
Under $X^n \sim \Bernoulli(\alpha^n)$ 
independently for $n\in[1:N]$, the initial conditions are given by  
\begin{equation}
	q(0,n) = \mathbb{P}\left(\sum_{m=n}^NX^m=0\right) = \prod_{m=n}^N\mathbb{P}(X^m=0) = \prod_{m=n}^N (1 - \alpha^m), \quad n\in[1:N],
\end{equation}
in the case of no success, $q(1,N)=\mathbb{P}(X^N=1)=\alpha^N$
where the sum reduces to a single Bernoulli variable, and 
$q(i,n)=0$ for $i>N-n+1$ because a sum of $N-n+1$ Bernoulli variables
cannot be larger than $N-n+1$, in particular $q(i,N)=0$ for all $i\geq 2$. 
For $i\in[1:N]$ and $n\in[1:N-1]$, by conditioning on the value of $X^{n}\in\{0,1\}$, the law of total probability gives 
\eqref{eqn:q-recursion}:
\begin{align}
	q(i,n) &= \mathbb{P}(X^{n}=1)~\mathbb{P}\left(\sum_{m=n}^NX^m=i \mid X^{n}=1\right) ~+~ 
	\mathbb{P}(X^{n}=0)~\mathbb{P}\left(\sum_{m=n}^NX^m=i \mid X^{n}=0\right) \notag\\
	&= \alpha^n~\mathbb{P}\left(\sum_{m=n+1}^NX^m=i-1\right) ~+~ 
	(1-\alpha^n)~\mathbb{P}\left(\sum_{m=n+1}^NX^m=i \right) \\
	&= \alpha^n q(i-1,n+1) + (1-\alpha^n) q(i,n+1). \notag
\end{align}

\subsection{A thinning result}
\label{sub:thinning}
We show that under the static model $X^n \sim \Bernoulli(\alpha^n)$ 
independently for $n\in[1:N]$ and the observation model 
$Y \mid X = x \sim \Binomial(I(x), \rho)$, 
we have $Y\sim\PB(\rho\,\alpha)$ marginally. 
We first note that the
characteristic function of $I(X)=\sum_{n=1}^N X^n \sim\PB(\alpha)$ is given by 
\begin{equation}\label{eqn:poibin_cf}
	\E\left[\exp(i\omega I(X))\right] = \prod_{n=1}^N\E\left[\exp(i\omega X^n)\right] = \prod_{n=1}^N \{\alpha^n\exp(i\omega) + (1-\alpha^n)\}
\end{equation}
for $\omega\in\mathbb{R}$. 
Consider the representation $Y=\sum_{n=1}^IZ^n$ where 
$Z^n\sim\Bernoulli(\rho)$ independently.
Note that the characteristic function of each Bernoulli random variable with success probability $\rho$ is 
$\varphi_Z(\omega) = \E\left[\exp(i\omega Z^n)\right]=\rho\exp(i\omega)+(1-\rho)$.
By the law of iterated expectations, the characteristic function of $Y$ is 
\begin{align}
		\E\left[\exp(i\omega Y)\right] 
		& = \E \left[ \E\left[\exp\left(i \omega \sum_{n=1}^IZ^n\right) \mid I \right] \right] \notag\\
		& = \E\left[\varphi_Z(\omega)^I\right] \notag\\
		& = \prod_{n=1}^N\E\left[\varphi_Z(\omega)^{X^n}\right] \\
		& = \prod_{n=1}^N\{\alpha^n \varphi_Z(\omega) + (1-\alpha^n)\} \notag\\
		& = \prod_{n=1}^N\{\rho\alpha^n\exp(i\omega) + (1-\rho\alpha^n)\}\notag
\end{align}
for $\omega\in\mathbb{R}$. 
By comparing this characteristic function with \eqref{eqn:poibin_cf}, we can 
conclude that $Y\sim\PB(\rho\,\alpha)$. 

\subsection{Comparing two Poisson binomial distributions}\label{sec:compare_poibin}
The following result will be of use in Appendix \ref{sec:csmc_appendix}. 

\begin{lemma}\label{lemma:l2_poibin}
Let $\PB(\alpha)$ and $\PB(\bar{\alpha})$ denote two Poisson binomial distributions 
with probabilities $\alpha=(\alpha^n)_{n\in[1:N]}$ and $\bar{\alpha}=(\bar{\alpha}^n)_{n\in[1:N]}$, respectively. 
The $\ell^2$-norm between these PMFs satisfies 
\begin{align}\label{eqn:l2_poibin}
	\sum_{i\in[0:N]}\left\{ \PB(i;\bar{\alpha}) - \PB(i;\alpha) \right\}^2 \leq \left(\sum_{n=1}^N|\bar{\alpha}^n-\alpha^n|\right)^2.
\end{align}
The Kullback--Leibler divergence from $\PB(\alpha)$ to $\PB(\bar{\alpha})$, defined as 
\begin{align*}
	\KL\left(\PB(\bar{\alpha}) \mid \PB(\alpha)\right) = \sum_{i\in[0:N]}\PB(i;\bar{\alpha}) \log\left(\frac{\PB(i;\bar{\alpha})} {\PB(i;\alpha)}\right),
\end{align*} 
is upper bounded by 
\begin{align}\label{eqn:KL_poibin}
	\KL\left(\PB(\bar{\alpha}) \mid \PB(\alpha)\right)\leq \left(\sum_{i\in[0:N]}\frac{\PB(i;\bar{\alpha})^2} {\PB(i;\alpha)^2}\right)^{1/2}
	\left(\sum_{n=1}^N|\bar{\alpha}^n-\alpha^n|\right).
\end{align}
\end{lemma}

\begin{proof}
To bound the $\ell^2$-norm between two Poisson binomial PMFs, we rely on Parseval's identity%\citep{parseval1806memoires}
\begin{align}\label{eqn:parseval}
\sum_{i\in[0:N]}\left\{ \PB(i;\bar{\alpha}) - \PB(i;\alpha) \right\}^2
=\frac{1}{2\pi}\int_{-\pi}^{\pi}\left\{ \bar{\varphi}(\omega;\bar{\alpha})-\varphi(\omega;\alpha\right\}^{2}d\omega,	
\end{align}
where 
\begin{equation}\label{eqn:charfuncpoibins}
	\bar{\varphi}(\omega;\bar{\alpha}) = \prod_{n=1}^N \{\bar{\alpha}^n\exp(i\omega) + (1-\bar{\alpha}^n)\},\quad
	\varphi(\omega;\alpha) = \prod_{n=1}^N \{\alpha^n\exp(i\omega) + (1-\alpha^n)\},
\end{equation}
denote the characteristic functions of $\PB(\bar{\alpha})$ and $\PB(\alpha)$, respectively
(as in \eqref{eqn:poibin_cf}).
Noting that each term of the products in \eqref{eqn:charfuncpoibins} is at most one, by repeated applications of 
triangle inequality on the decomposition 
\begin{align}
	&\bar{\varphi}(\omega;\bar{\alpha}) - \varphi(\omega;\alpha) \\
	&= \sum_{n=1}^N (\bar{\alpha}^n-\alpha^n) (\exp(i\omega)-1)
	\prod_{m=1}^{n-1}\{\alpha^m\exp(i\omega) + (1-\alpha^m)\}
	\prod_{p=n+1}^N\{\bar{\alpha}^p\exp(i\omega) + (1-\bar{\alpha}^p)\}\notag
\end{align}
(with the convention $\prod_{n=m}^p=1$ for $p<m$), we have
\begin{align}
	\left\{\frac{1}{2\pi}\int_{-\pi}^{\pi}\left\{ \bar{\varphi}(\omega;\bar{\alpha})-\varphi(\omega;\alpha)\right\}^{2}d\omega\right\}^{1/2}
	\leq \sum_{n=1}^N \left\{ \frac{1}{2\pi}\int_{-\pi}^{\pi} (\bar{\alpha}^n-\alpha^n)^2 (\exp(i\omega)-1)^2d\omega \right\}^{1/2}
	=\sum_{n=1}^N|\bar{\alpha}^n-\alpha^n|\notag.
\end{align}
Hence \eqref{eqn:l2_poibin} follows by squaring both sides and applying the identity in \eqref{eqn:parseval}. \\

To bound the Kullback--Leibler divergence, we apply the inequality $\log(x)\leq x-1$ for $x>0$ and the Cauchy--Schwarz inequality 
\begin{align}
	&\KL\left(\PB(\bar{\alpha}) \mid \PB(\alpha)\right)
	\leq \sum_{i\in[0:N]}\left\{\frac{\PB(i;\bar{\alpha})} {\PB(i;\alpha)}\right\} \left\{ \PB(i;\bar{\alpha}) - \PB(i;\alpha) \right\}\notag\\
	&\leq \left(\sum_{i\in[0:N]}\frac{\PB(i;\bar{\alpha})^2} {\PB(i;\alpha)^2}\right)^{1/2}
	\left(\sum_{i\in[0:N]}\left\{ \PB(i;\bar{\alpha}) - \PB(i;\alpha) \right\}^2\right)^{1/2}.
\end{align}
Applying the inequality in \eqref{eqn:l2_poibin} leads to \eqref{eqn:KL_poibin}.

\end{proof}

\section{Controlled sequential Monte Carlo}\label{sec:csmc_appendix}
\subsection{Performance of controlled sequential Monte Carlo}\label{sec:KL_proposal}
\begin{proof}[Proof of Proposition \ref{prop:KL_BIF_bound}]
Using the form of the smoothing distribution $p_{\theta}(x_{0:T}|y_{0:T})$ 
in \eqref{eqn:smoothing_distribution_bif} and \eqref{eqn:decompose_smoothing}, 
and the definition of the proposal distribution $q_{\theta}(x_{0:T})$ 
in \eqref{eqn:csmc_proposals}, we have 
\begin{align}\label{eqn:logratio_decomp}
	\log\left(\frac{p_{\theta}(x_{0:T}|y_{0:T})}{q_{\theta}(x_{0:T})}\right) = 
  \log\left(\frac{\mu_{\theta}(\psi_0)}{\mu_{\theta}(\psi_0^\star)}\right) + 
	\sum_{t=0}^T\log\left(\frac{\psi_t^\star(x_t)}{\psi_t(I(x_t))}\right) + 
	\sum_{t=1}^T\log\left(\frac{f_{\theta}(\psi_t|x_{t-1})}{f_{\theta}(\psi_t^\star|x_{t-1})}\right).
\end{align}
By the log-sum inequality 
\begin{align}\label{eqn:LSI_KL1}
	\log\left(\frac{\mu_{\theta}(\psi_0)}{\mu_{\theta}(\psi_0^\star)}\right) 
	&\leq \mu_{\theta}(\psi_0)^{-1}\sum_{x_0\in\{0,1\}^N}\mu_{\theta}(x_0)\psi_0(x_0)
	\log\left(\frac{\psi_0(I(x_0))}{\psi_0^\star(x_0)}\right)
	\leq \eta_0(\log(\psi_0/\psi_0^\star)|\theta)
\end{align}
and 
\begin{align}\label{eqn:LSI_KL2}
	\log\left(\frac{f_{\theta}(\psi_t|x_{t-1})}{f_{\theta}(\psi_t^\star|x_{t-1})}\right)
	&\leq f_{\theta}(\psi_t|x_{t-1})^{-1}\sum_{x_t\in\{0,1\}^N}f_{\theta}(x_t|x_{t-1})\psi_t(I(x_t))
	\log\left(\frac{\psi_t(I(x_t))}{\psi_t^\star(x_t)}\right)\notag\\
	&\leq q_t(\log(\psi_t/\psi_t^\star)|x_{t-1},\theta).
\end{align}
Using the expression in \eqref{eqn:logratio_decomp} and the inequalities \eqref{eqn:LSI_KL1}-\eqref{eqn:LSI_KL2}, 
the Kullback--Leibler divergence from $q_{\theta}(x_{0:T})$ to $p_{\theta}(x_{0:T}|y_{0:T})$ satisfies 
\begin{align}
	&\KL\left(p_{\theta}(x_{0:T}|y_{0:T}) \mid q_{\theta}(x_{0:T})\right) 
	= \sum_{x_{0:T}\in\{0,1\}^{N\times(T+1)}} p_{\theta}(x_{0:T}|y_{0:T}) \log\left(\frac{p_{\theta}(x_{0:T}|y_{0:T})}{q_{\theta}(x_{0:T})}\right)\\
	&\leq  
  \eta_0(\log(\psi_0/\psi_0^\star)|\theta)
  +	\sum_{t=0}^T \eta_{t}^{\star}(\log(\psi_t^{\star}/\psi_t)|\theta) + \sum_{t=1}^T\sum_{x_{t-1}\in\{0,1\}^{N}}\eta_{t-1}^\star(x_{t-1}|\theta)q_t(\log(\psi_t/\psi_t^\star)|x_{t-1},\theta).\notag
\end{align}
Equation \eqref{eqn:KL_BIF_bound} follows by employing the upper bound 
$M_t=\max_{x_{t}\in\{0,1\}^N}\eta_{t}^\star(x_{t}|\theta)/\eta_{t}(x_{t}|\theta)$ 
and noting that 
\begin{align}
	\sum_{x_{t-1}\in\{0,1\}^{N}}\eta_{t-1}(x_{t-1}|\theta)q_t(\log(\psi_t/\psi_t^\star)|x_{t-1},\theta) 
	= \eta_t(\log(\psi_t/\psi_t^\star)|\theta).
\end{align}
\end{proof}

\subsection{Marginal distributions of smoothing distribution and proposal distribution}\label{sec:ratio_marginal}
To show that the constants 
$M_t=\max_{x_{t}\in\{0,1\}^N}\eta_{t}^\star(x_{t}|\theta)/\eta_{t}(x_{t}|\theta)$ 
in Proposition \ref{prop:KL_BIF_bound} are finite for all $t\in[0:T-1]$, 
it suffices to upper bound the ratio 
$p_{\theta}(x_{0:T}|y_{0:T}) / q_{\theta}(x_{0:T})$ by a constant for all $x_{0:T}\in\{0,1\}^{N(T+1)}$. 
Under the requirement \eqref{eqn:weights_requirement}, we have 
\begin{align}
	\frac{p_{\theta}(x_{0:T}|y_{0:T})}{q_{\theta}(x_{0:T})} = p_{\theta}(y_{0:T})^{-1}\prod_{t=0}^{T-1}w_t(x_t), 
\end{align}
hence we will argue that each of the above weight functions can be upper bounded by some constant. 
From \eqref{eqn:csmc_weights}, the weight functions can be rewritten as 
$w_0(x_0) = \mu_{\theta}(\psi_0)f_{\theta}(\psi_1|x_0)/\bar{f}_{\theta}(\psi_0| I(x_0))$ and 
$w_t(x_t) = f_{\theta}(\psi_{t+1}|x_t)/\bar{f}_{\theta}(\psi_{t+1}| I(x_t))$ for $t\in[1:T-1]$. 
Noting that the ratio $f_{\theta}(\psi_{t+1}|x_t)/\bar{f}_{\theta}(\psi_{t+1}| I(x_t))=1$ 
when $x_t^n=0$ for all $n\in[1:N]$ and $t\in[0:T-1]$, we do not have to consider this case. 

By induction, the BIF approximation $(\psi_t)_{t\in[0:T]}$ are upper bounded by one, 
hence $\mu_{\theta}(\psi_0)\leq1$ and $f_{\theta}(\psi_{t}|x_{t-1})\leq 1$ for all $x_{t-1}\in\{0,1\}^N$ and $t\in[1:T]$.
It remains to show that $\bar{f}_{\theta}(\psi_{t+1}| I(x_t))$ is lower bounded by some strictly positive constant. 
We notice that the conditional expectation in \eqref{eqn:sis_condexp_approx} can be lower bounded by 
\begin{align}
	\bar{f}_{\theta}(\psi_{t+1} | I(x_{t})) \geq 
	 \left(\bar{\lambda}N^{-1}I(x_t)\right)^{N-I(x_t)} \left(1-\bar{\gamma}\right)^{I(x_t)} \psi_{t+1}(N).
\end{align}
Define the constant $c = \min_{i \in [1:N]} \left(\bar{\lambda}N^{-1}i\right)^{N-i} \left(1-\bar{\gamma}\right)^{i}>0$. 
At the terminal time, we have $\psi_T(N) = \Binomial(y_T;N,\rho)>0$. 
By iterating the backward recursion, we have
\begin{align}
	\psi_{t}(N) = \Binomial(y_t;N,\rho) \bar{f}_{\theta}(\psi_{t+1}| N) \geq \Binomial(y_t;N,\rho) \,c\, \psi_{t+1}(N),
\end{align}
for each $t\in[0:T-1]$. By induction, we have 
$\bar{f}_{\theta}(\psi_{t+1} | I(x_{t})) \geq c^{T-t} \prod_{k=t+1}^T \Binomial(y_{k}; N,\rho)$ for all $t \in [0:T-1]$ 
and $x_t\in\{0,1\}^N$. 
Combining the above observations allows us to conclude.

\subsection{Backward information filter approximation}
\label{sec:bif_approx}
\begin{proof}[Proof of Proposition \ref{prop:BIF_error}]
We will establish \eqref{eqn:BIF_logerror1};
the proof of \eqref{eqn:BIF_logerror2} follows using similar arguments. 
Define $e_t=\eta_{t}^{\star}(\log(\psi_t^{\star}/\psi_t)|\theta)$ for each time $t\in[0:T]$. 
At the terminal time $T$, we have $e_T=0$. 
For time $t\in[0:T-1]$, we consider the decomposition 
\begin{align}\label{eqn:BIF_logdecomp}
	e_t = \eta_{t}^{\star}(\log(\psi_t^{\star}/\tilde{\psi}_t)|\theta) + 
	\eta_{t}^{\star}(\log(\tilde{\psi}_t/\psi_t)|\theta)
\end{align}
where
\begin{align}\label{eqn:BIF_intermediate}
	\tilde{\psi}_{t}(x_{t})=\Binomial(y_{t};I(x_{t}),\rho)\mathbbm{1}(I(x_{t})\geq y_{t})f_{\theta}(\psi_{t+1}|x_{t}).
\end{align}
Using the log-sum inequality and the decomposition of the smoothing distribution given 
in \eqref{eqn:smoothing_distribution_bif} and \eqref{eqn:decompose_smoothing}, 
the first term of \eqref{eqn:BIF_logdecomp} is bounded by 
\begin{align}\label{eqn:BIF_logterm_1}
	 &\eta_{t}^{\star}(\log(\psi_t^{\star}/\tilde{\psi}_t)|\theta)
	= \sum_{x_t\in\{0,1\}^N} \log\left(\frac{f_{\theta}(\psi_{t+1}^\star|x_t)}
	 {f_{\theta}(\psi_{t+1}|x_t)}\right) \eta_t^\star(x_t|\theta)\notag\\
	 &\leq \sum_{x_t\in\{0,1\}^N} f_{\theta}(\psi_{t+1}^\star|x_t)^{-1} 
	 \sum_{x_{t+1}\in\{0,1\}^N}f_{\theta}(x_{t+1}|x_t)\psi_{t+1}^\star(x_{t+1}) \log\left(\frac{\psi_{t+1}^\star(x_{t+1})}{\psi_{t+1}(I(x_{t+1}))}\right) 	 
	 \eta_{t}^\star(x_t|\theta)\notag\\
	 &\leq \sum_{x_{t}\in\{0,1\}^N}\sum_{x_{t+1}\in\{0,1\}^N}
	 \log\left(\frac{\psi_{t+1}^\star(x_{t+1})}{\psi_{t+1}(I(x_{t+1}))}\right) \eta_{t}^\star(x_t|\theta) p_{\theta}(x_{t+1}|x_t,y_{t+1:T})
	 = e_{t+1}.
\end{align}
By the log-sum inequality, the second term of \eqref{eqn:BIF_logdecomp} is bounded by 
\begin{align}\label{eqn:sis_logerror_first}
	&\eta_{t}^{\star}(\log(\tilde{\psi}_t/\psi_t)|\theta)
	= \sum_{x_t\in\{0,1\}^N\setminus (0,\ldots,0)} \log\left(\frac{f_{\theta}(\psi_{t+1}|x_t)}{\bar{f}_{\theta}(\psi_{t+1}|I(x_t))}\right) \eta_t^\star(x_t|\theta)\notag\\
	 &\leq \sum_{x_{t}\in\{0,1\}^N\setminus (0,\ldots,0)} f_{\theta}(\psi_{t+1}|x_t)^{-1}
	 \sum_{i_{t+1}=0}^N 
	 \PB(i_{t+1};\alpha(x_t)) \psi_{t+1}(i_{t+1}) \log \left(\frac{\PB(i_{t+1};\alpha(x_t))}{\PB(i_{t+1};\bar{\alpha}(x_t))}\right) 
	 \eta_t^\star(x_t|\theta)\notag\\
	 &\leq c_{\theta}^\star(\psi_{t+1}) \sum_{x_{t}\in\{0,1\}^N}\KL\left(\PB(\alpha(x_t)) \mid 
	 \PB(\bar{\alpha}(x_t)) \right) \eta_t^\star(x_t|\theta).
\end{align}
The above repeatedly uses the fact that in the case of zero infections, i.e. $x_t^n=0$ for all $n\in[1:N]$, 
we have $\alpha^n(x_t)= \bar{\alpha}^n(x_t)=0$ for all $n\in[1:N]$. 
The constant $c_{\theta}^\star(\psi_{t+1})$ can be shown to be finite using the arguments in Appendix \ref{sec:ratio_marginal}. 
The last inequality also uses the fact that $\max_{i_{t+1}\in[0:N]} |\psi_{t+1}(i_{t+1})|\leq 1$, which follows by induction. 
Applying \eqref{eqn:KL_poibin} of Lemma \ref{lemma:l2_poibin} and Cauchy--Schwarz inequality gives 
\begin{align}\label{eqn:BIF_logterm_2}
	\eta_{t}^{\star}(\log(\tilde{\psi}_t/\psi_t)|\theta)\leq c_{\theta}^\star(\psi_{t+1}) 
	\|\phi_{\theta}^\star\|_{L^2(\eta_t^\star)}\|\Delta_{\theta}\|_{L^2(\eta_t^\star)}.
\end{align}
By combining \eqref{eqn:BIF_logterm_1} and \eqref{eqn:BIF_logterm_2}, we obtain the following recursion 
\begin{align}
	e_t \leq e_{t+1} + c_{\theta}^\star(\psi_{t+1}) 
	\|\phi_{\theta}^\star\|_{L^2(\eta_t^\star)}\|\Delta_{\theta}\|_{L^2(\eta_t^\star)} . 
\end{align}
The claim in \eqref{eqn:BIF_logerror1} follows by induction.
\end{proof}

\begin{proof}[Proof of Proposition \ref{prop:sir_BIF_error}]
The arguments are similar to Proposition \ref{prop:BIF_error} 
with some modifications that we will outline below. 
Like before, we consider the quantity 
$e_t = \eta_t^{\star}(\log (\psi^{\star}_t / \psi_t) | \theta)$ for $t\in[0:T]$ 
in \eqref{eqn:sir_BIF_logerror1} as 
\eqref{eqn:sir_BIF_logerror2} follows using similar arguments. 
By the same arguments in \eqref{eqn:BIF_logterm_1} 
\begin{align}
	\eta_{t}^{\star}(\log(\psi_t^{\star}/\tilde{\psi}_t)|\theta) \leq e_{t+1}.
\end{align}
Instead of \eqref{eqn:sis_logerror_first}, we have 
\begin{align}
\eta_{t}^{\star}(\log(\tilde{\psi}_t/\psi_t)|\theta)
  \leq c_{\theta}^{\star}(\psi_{t+1}) 
  \sum_{x_t\in\{0,1,2\}^N}
  \KL\left({f}_{\theta}(s_{t+1},i_{t+1}|x_t) \mid \bar{f}_{\theta}(s_{t+1},i_{t+1}|S(x_t),I(x_t)) \right) \eta_t^\star(x_t|\theta).
\end{align}
The Kullback--Leibler divergence from $\bar{f}_{\theta}(s_{t+1},i_{t+1}|S(x_t),I(x_t))$ to ${f}_{\theta}(s_{t+1},i_{t+1}|x_t)$ can be decomposed as 
\begin{align}
  &\KL\left({f}_{\theta}(s_{t+1},i_{t+1}|x_t) \mid \bar{f}_{\theta}(s_{t+1},i_{t+1}|S(x_t),I(x_t)) \right)\\  
  &= \KL\left({f}_{\theta}(s_{t+1}|x_t) \mid \bar{f}_{\theta}(s_{t+1}|S(x_t),I(x_t)) \right) \notag \\
  & \quad \quad + \sum_{s_{t+1}=0}^N{f}_{\theta}(s_{t+1}|x_t)\KL\left({f}_{\theta}(i_{t+1}|x_t, s_{t+1}) \mid \bar{f}_{\theta}(i_{t+1}|S(x_t),I(x_t),s_{t+1}) \right)\notag\\
  &\leq\KL\left( \PB({\alpha}_S(x_t)) \mid \PB(\bar\alpha_S(x_t)) \right) + 
  \KL\left( \PB((I^n(x_t){\gamma})_{n\in[1:N]}) \mid 
  \PB((I^n(x_t)\bar\gamma^n)_{n\in[1:N]}) \right)\notag.
\end{align}
Applying \eqref{eqn:KL_poibin} of Lemma \ref{lemma:l2_poibin} gives 
\begin{align}
  &\KL\left({f}_{\theta}(s_{t+1},i_{t+1}|x_t) \mid \bar{f}_{\theta}(s_{t+1},i_{t+1}|S(x_t),I(x_t)) \right)  
  \leq \phi^{\star}_{\theta}(x_t)\Delta_{\theta,S}(x_t) + \xi^{\star}_{\theta}(x_t)\Delta_{\theta,R}(x_t). 
\end{align}
Hence by Cauchy--Schwarz inequality, we have
\begin{align}
\eta_{t}^{\star}(\log(\tilde{\psi}_t/\psi_t)|\theta)
  \leq c_{\theta}^{\star}(\psi_{t+1}) 
  \{\|\phi^{\star}_{\theta}\|_{L^2(\eta_t^\star)}\|\Delta_{\theta,S}\|_{L^2(\eta_t^\star)}
  + \|\xi^{\star}_{\theta}\|_{L^2(\eta_t^\star)}\|\Delta_{\theta,R}\|_{L^2(\eta_t^\star)}\}.
\end{align}
\end{proof}

\subsection{Finer-grained approximations}\label{sec:finer_BIF_approx}
We discuss finer approximations of the BIF and provide algorithmic details of the resulting cSMC. 
Instead of simply approximating individual infection and recovery rates by their population averages, 
we consider a clustering of agents based on their infection and recovery rates, and approximate 
these rates by their within-cluster average. 
More precisely, if $K\in[1:N]$ denote the desired number of clusters,  
$(C_k)_{k\in[1:K]}$ denote a partition of $[1:N]$ that 
represents our clustering and $N_k=|C_k|$ denotes the number of agents in cluster $k\in[1:K]$, 
we approximate $\lambda^n\approx\bar{\lambda}^k=N_k^{-1}\sum_{n\in C_k}\lambda^n$ and  
$\gamma^n\approx\bar{\gamma}^k=N_k^{-1}\sum_{n\in C_k}\gamma^n$. 
For each state of the population $x\in\{0,1\}^N$, we define the summary 
$I^k(x)=\sum_{n\in C_k}x^n$, which counts the number of infected agents in cluster $k\in[1:K]$. 
We will write $x^{C_k}=(x^n)_{n\in C_k}$ and $I^{1:K}(x)=(I^k(x))_{k\in[1:K]}$. Note that $\sum_{k=1}^KI^k(x) = I(x)$. 
The corresponding approximation of $\alpha(x)$ is given by 
\begin{equation}
	\label{eqn:sis_alpha_finerapprox}
		\bar{\alpha}^n(x) = \begin{cases}
					\bar{\lambda}^{c(n)} N^{-1}I(x),\quad & \mbox{if }x^n=0,\\
					1-\bar{\gamma}^{c(n)},\quad & \mbox{if }x^n=1,
					\end{cases}
\end{equation}
where $c:[1:N]\rightarrow[1:K]$ maps agent labels to cluster membership. 
We will write associated Markov transition as $\bar{f}_{\theta}(x_t|x_{t-1})=\prod_{n=1}^N\Bernoulli(x_{t}^n;\bar{\alpha}^n(x_{t -1}))$.
The approximation in \eqref{eqn:sis_alpha_finerapprox} recovers \eqref{eqn:sis_alpha_approx} in the case of $K=1$ cluster, 
in which case $I^1(x)=I(x)\in[0:N]$. 
Larger values of $K$ can be seen as finer-grained approximations at the expense of increased dimensionality of 
$I^{1:K}(x)\in\prod_{k=1}^K[0:N_k]$. Having $K=N$ clusters offers no dimensionality reduction as $I^{1:N}(x)=x\in\{0,1\}^N$.

At the terminal time $T$, since we have $\psi_T^\star(x_T)=\Binomial(y_T;I(x_T), \rho)\mathbbm{1}(I(x_T) \geq y_T)$, 
it suffices to compute $\psi_T(i_T^{1:K})=\Binomial(y_T;\sum_{k=1}^Ki_T^k, \rho)\mathbbm{1}(\sum_{k=1}^Ki_T^k \geq y_T)$ 
for all $i_T^{1:K}\in\prod_{k=1}^K[0:N_k]$. 
To approximate the backward recursion \eqref{eqn:backward_info_filter} for $t\in[0:T-1]$, 
we proceed inductively by assuming an approximation of the form $\psi_{t+1}(I^{1:K}(x_{t+1}))$ at time $t+1$. 
By plugging in the approximation $\psi_{t+1}(I^{1:K}(x_{t+1}))\approx \psi_{t+1}^\star(x_{t+1})$, we consider 
the iterate
\begin{align}
	\Binomial(y_{t};I(x_{t}), \rho)\mathbbm{1}(I(x_{t}) \geq y_{t})f_{\theta}(\psi_{t+1}|x_{t})
\end{align}
to form an approximation of $\psi_t^\star(x_t)$. We approximate the conditional expectation $f_{\theta}(\psi_{t+1}|x_{t})$ by 
\begin{align}
  \label{eqn:sir_approx_conditional_expectation}
	\bar{f}_{\theta}(\psi_{t+1}|I^{1:K}(x_{t})) = \sum_{i_{t+1}^{1:K}\in\prod_{k=1}^K[0:N_k]}
	\prod_{k=1}^K\SB(i_{t+1}^k;N_{k}-I^{k}(x_{t}), \bar{\lambda}^{k}N^{-1}I(x_{t}), 
	I^{k}(x_{t}), 1-\bar{\gamma}^{k})\psi_{t+1}(i_{t+1}^{1:K}).
\end{align}
The resulting approximation of $\psi_{t}^\star(x_{t})$ is computed using 
\begin{align}\label{eqn:BIFapprox_finer}
	\psi_{t}(i_{t}^{1:K}) = \Binomial\left(y_{t};\sum_{k=1}^Ki_{t}^k, \rho\right)
	\mathbbm{1}\left(\sum_{k=1}^Ki_{t}^k \geq y_{t}\right)\bar{f}_{\theta}(\psi_{t+1}|i_{t}^{1:K})
\end{align}
for all $i_{t}^{1:K}\in\prod_{k=1}^K[0:N_k]$ in $\bigo\left(\prod_{k=1}^K(1+N_k)^2\right)$ cost. 
Hence our overall approximation of the BIF $(\psi_t)_{t\in[0:T]}$ costs
$\bigo(\prod_{k=1}^K(1+N_k)^2 T)$ to compute and
$\bigo(\prod_{k=1}^K(1+N_k) T  )$ in storage.  It is straightforward
to extend Proposition \ref{prop:BIF_error} to characterize the error of
\eqref{eqn:BIFapprox_finer}; we omit this for the sake of brevity. 

The corresponding approximation of the optimal proposal \eqref{eqn:decompose_smoothing} is
\begin{align}\label{eqn:csmc_proposals_finer}
	q_0(x_0|\theta) = \frac{\mu_{\theta}(x_0)\psi_0(I^{1:K}(x_0))}{\mu_{\theta}(\psi_0)}, 
	\quad q_t(x_t|x_{t-1},\theta) = \frac{f_{\theta}(x_t|x_{t-1})\psi_t(I^{1:K}(x_t))}{f_{\theta}(\psi_t|x_{t-1})}, \quad t\in[1:T].
\end{align}
Analogous to \eqref{eqn:csmc_weights}, the appropriate SMC weight functions are
\begin{align}\label{eqn:csmc_weights_finer}
	&w_0(x_0) = \frac{\mu_{\theta}(\psi_0)g_{\theta}(y_0|x_0)f_{\theta}(\psi_1|x_0)}{\psi_0(I^{1:K}(x_0))},
	&w_t(x_t) = \frac{g_{\theta}(y_t|x_t)f_{\theta}(\psi_{t+1}|x_t)}{\psi_t(I^{1:K}(x_t))},\quad t\in[1:T-1],
\end{align}
and $w_T(x_T) = 1$. The above expectations are computed as 
\begin{align}
	\mu_{\theta}(\psi_0) &= \sum_{i_{0}^{1:K}\in\prod_{k=1}^K[0:N_k]}
	\prod_{k=1}^K\PB(i_0^k;\alpha_0^{C_k})\psi_0(i_0^{1:K}),\\
	f_{\theta}(\psi_{t}|x_{t-1}) &= \sum_{i_{0}^{1:K}\in\prod_{k=1}^K[0:N_k]}
	\prod_{k=1}^K\PB(i_t^k;\alpha^{C_k}(x_{t-1}))\psi_t(i_t^{1:K}),\quad t\in[1:T],
\end{align}
where $\alpha_0^{C_k}=(\alpha_0^n)_{n\in C_k}$ and $\alpha^{C_k}(x_{t-1})=(\alpha^n(x_{t-1}))_{n\in C_k}$.
Sampling from the proposals \eqref{eqn:csmc_proposals_finer} involves initializing from 
\begin{align}
	q_0(x_0,i_0^{1:K}|\theta) = q_0(i_0^{1:K}|\theta)\prod_{k=1}^Kq_0(x_0^{C_k}|i_0^k,\theta) 
	= \frac{\prod_{k=1}^K\PB(i_0^k;\alpha_0^{C_k})\psi_0(i_0^{1:K})}{\mu_{\theta}(\psi_0)}
	\prod_{k=1}^K\CB(x_0^{C_k};\alpha_0^{C_k}, i_0^k) 
\end{align}
which admits $q_0(x_0|\theta)$ as its marginal distribution, and for time $t\in[1:T]$ sampling from
\begin{align}
	q_t(x_t,i_t^{1:K}|x_{t-1},\theta) & = q_t(i_t^{1:K}|x_{t-1},\theta)\prod_{k=1}^Kq_t(x_t^{C_k}|x_{t-1},i_t^k,\theta)\\
	& = \frac{\prod_{k=1}^K\PB(i_t^k;\alpha^{C_k}(x_{t-1}))\psi_t(i_t^{1:K})}{f_{\theta}(\psi_{t}|x_{t-1})}
	\prod_{k=1}^K\CB(x_t^{C_k};\alpha^{C_k}(x_{t-1}), i_t^k)\notag
\end{align}
which admits $q_t(x_t|x_{t-1},\theta)$ as its marginal transition. 
The overall cost of implementing cSMC is 
\begin{align}
	\bigo\left(\left(\sum_{k=1}^KN_k^2+\prod_{k=1}^KN_k\right)TP\right),
\end{align}
or 
\begin{align}
	\bigo\left(\left(\sum_{k=1}^KN_k\log(N_k)+\prod_{k=1}^KN_k\right)TP\right),
\end{align}
if one employs translated Poisson approximations \eqref{eqn:transpoi} and the MCMC in~\citet{heng2020simple}
to sample from conditioned Bernoulli distributions. 
As alluded earlier, the choice of the number of clusters $K$ allows one to trade-off the quality of our proposal 
for computation complexity. 

\section{Posterior sampling of agent states using MCMC}\label{sub:mcmc}
We now describe various MCMC algorithms that can be used to sample from the 
smoothing distribution $p_{\theta}(x_{0:T}|y_{0:T})$ of the agent-based SIS model in Section \ref{sec:sis}. 
These Gibbs sampling strategies can be seen as alternatives to the SMC approach 
in Sections \ref{sub:sis_apf} and \ref{sub:sis_controlled_smc}. 
The same ideas can also be applied to the SIR model in Section \ref{sec:sir} with some modifications. 

A simple option is a single-site update that samples the state of an agent at a specific time 
from the full conditional distribution $p_{\theta}(x_{t}^n|x_{0:t-1},x_t^{-n},x_{t+1:T},y_{0:T})$, 
where $x_t^{-n}=(x_t^m)_{m\in[1:N]\setminus n}$, for $t\in[0:T]$ and $n\in[1:N]$ that 
can be chosen deterministically or randomly with uniform probabilities using a systematic or random Gibbs scan. 
Using the conditional independence structure of the model, for $t\in[1:T-1]$, this full 
conditional is a Bernoulli distribution with success probability that is proportional to  
\begin{align}\label{eqn:sis_gibbs_conditional_distribution}
	\alpha^n(x_{t-1})\Binomial(y_t;I(\tilde{x}_t),\rho)\prod_{m\in\mathcal{N}(n)}
	\Bernoulli(x_{t+1}^m;\alpha^m(\tilde{x}_{t})),
\end{align}
where $\tilde{x}_t=(x_t^m)_{m\in[1:N]}$ with $x_t^n=1$. This expression also holds 
for the case $t=0$ by replacing $\alpha^n(x_{t-1})$ with $\alpha_0^n$, and $t=T$ 
by removing the last product. 
As each update costs $\bigo(\mathcal{D}(n))$, the overall cost of a Gibbs scan is 
 $\bigo(T\sum_{n=1}^N\mathcal{D}(n))$.

As discussed in \citet{ghahramani1996factorial} for a related class of models, 
one can consider block updates that jointly sample the trajectory of a few agents. 
For a subset $b \subseteq [1:N]$ of size $B$, this Gibbs update involves sampling from 
$p_{\theta}(x_{0:T}^b|x_{0:T}^{-b},y_{0:T})$, where $x_{0:T}^b=(x_t^n)_{t\in[0:T],n\in b}$ and 
$x_{0:T}^{-b}=(x_t^n)_{t\in[0:T],n\in[1:N]\setminus b}$. This can be done using the 
forward-backward algorithm in $\bigo((2^{2B}+2^{B}N)T)$ cost. 
Hence the overall cost of a systematic Gibbs scan over the entire population is $\bigo((2^{2B}+2^{B}N)TN/B)$. 
As the latter is computationally prohibitive for large $B$, one has to consider reasonably small block sizes. 

We can also employ a swap update that leaves the full conditional distribution 
$p_{\theta}(x_t | x_{0:t-1},x_{t+1:T},y_{0:T})$ invariant for each $t \in [0:T]$ 
(with $x_{0:-1}=x_{T+1:T}=\emptyset$).  
Let $x_t$ denote the current state of the population at time $t$.
We will choose $n_0$ and $n_1$ uniformly from the sets $\{i:x_t^i=0\}$ and $\{i:x_t^i = 1\}$ respectively. 
The proposed state $\tilde{x}_t$ is such that 
$\tilde{x}_t^{n_0}=1$, $\tilde{x}_t^{n_1}=0$ and $\tilde{x}_t^n=x_t^n$ for $n\in[1:N]\setminus\{n_0,n_1\}$.
For $t\in[1:T-1]$, the MH acceptance probability is 
\begin{align}
	\min\left\lbrace 1, \frac{\alpha^{n_0}(x_{t-1})(1-\alpha^{n_1}(x_{t-1}))\prod_{n=1}^N\Bernoulli(x_{t+1}^n;\alpha^n(\tilde{x}_t))}
	{\alpha^{n_1}(x_{t-1})(1-\alpha^{n_0}(x_{t-1}))\prod_{n=1}^N\Bernoulli(x_{t+1}^n;\alpha^n(x_t))} \right\rbrace.
\end{align}
The case $t=0$ can be accommodated by replacing $\alpha(x_{t-1})$ with $\alpha_0$, 
and $t=T$ by removing the ratios of Bernoulli PMFs. The cost of each swap update ranges between 
$\bigo(1)$ and $\bigo(N)$, depending on the structure of the network. 
As swap updates on their own do not change the number of infected agents, 
we propose to employ them within a mixture kernel that also relies on 
the above-mentioned forward-backward scans.

As alternatives to PMMH, we can consider Gibbs samplers that alternate between sampling the 
parameters from $p(\theta|x_{0:T},y_{0:T})$ using a MH algorithm, and the agent states from 
$p_{\theta}(x_{0:T}|y_{0:T})$ using the above MCMC algorithms. 
In particular, we employ an equally weighted mixture of kernels, 
performing either $N$ random swap updates for each time step, or 
systematic Gibbs scan with single-site updates (``single-site'') or block updates for $B=5$ agents (``block5''). 
We compare these Gibbs samplers to PMMH on the simulated dataset of Section~\ref{sub:sis_illustrate_smc}. 
We set the MH Normal proposal standard deviation within Gibbs samplers as $0.08$ to achieve suitable acceptance probabilities. 
We initialize the parameters from either the DGP or the prior distribution, and run all algorithms for the same number of iterations. 
Figure~\ref{fig:gibbs_compare_pmcmc_hist1} displays the resulting approximations of the marginal posterior distribution 
of two parameters. 
Compared to the Gibbs samplers, the PMMH approximations seem to be more stable between runs with different initializations. 
These plots also suggest that the Gibbs samplers provide a poorer exploration of the tails of the posterior distribution, 
which may be attributed to the dependence between the agent states and parameters in this setting. 

\begin{figure}[htbp]
  \centering
  \includegraphics[width=1.0\textwidth]{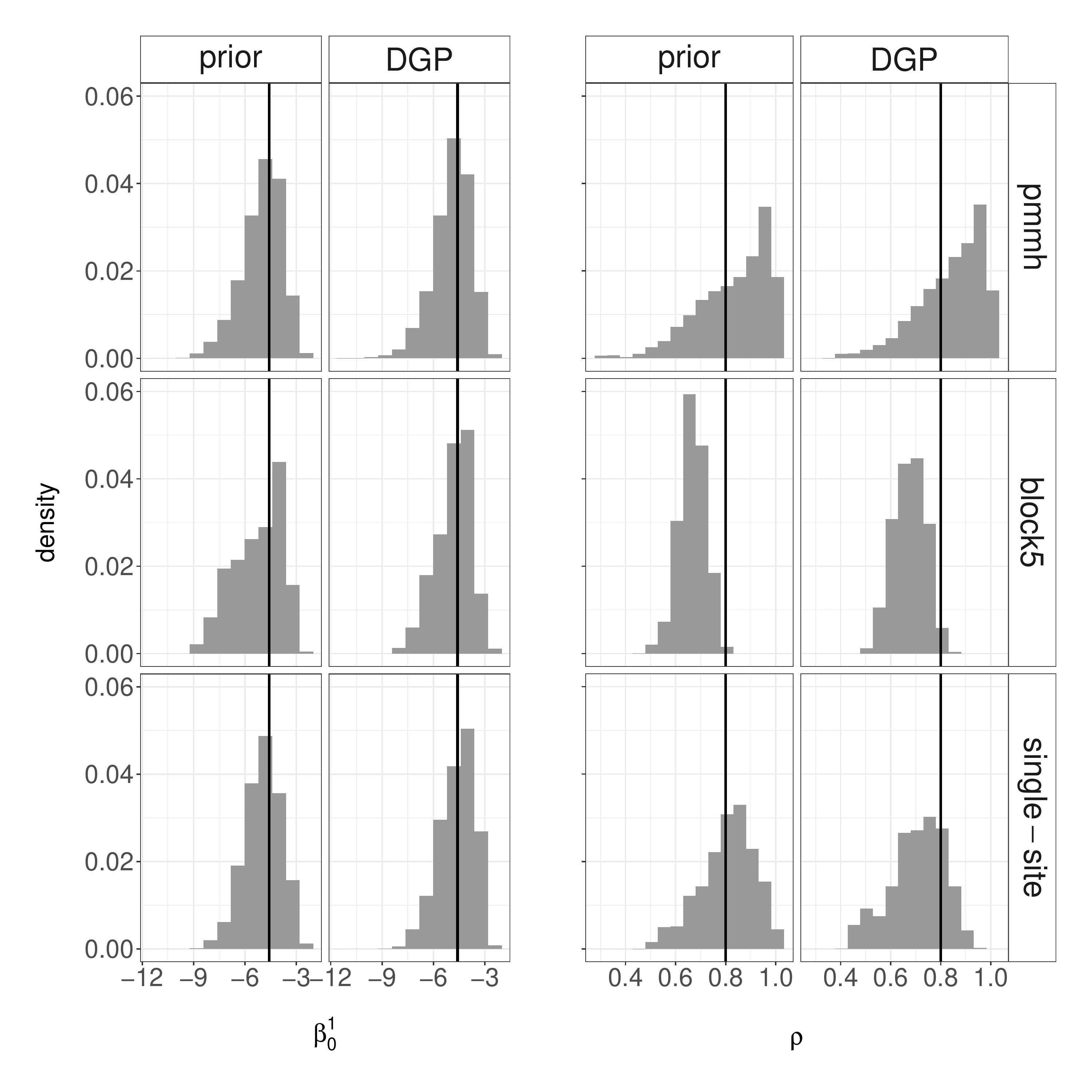}
  \caption{Histograms illustrating approximations of the marginal posterior distributions 
  of $\beta_0^1$ and $\rho$, obtaining using PMMH chains and Gibbs samplers with 
  either single-site or block updates, and initialization from either the DGP or the prior distribution. 
  Black vertical lines depict the DGP.}
  \label{fig:gibbs_compare_pmcmc_hist1}
\end{figure}

\end{document}